%% file: main_www2025.tex
\documentclass[sigconf]{acmart}

\input{main_common}

\usepackage[utf8]{inputenc} 
\usepackage[T1]{fontenc}    
\usepackage{hyperref}       
\usepackage{url}            
\usepackage{booktabs}       
\usepackage{balance}
\usepackage{nicefrac}       
\usepackage{microtype}      
\usepackage{xcolor}         
\usepackage{amsmath}
\usepackage{stfloats}
\DeclareMathOperator*{\argmax}{arg\,max}

\usepackage[a-2b]{pdfx}

\usepackage{subcaption}
\usepackage{graphicx}
\usepackage{xcolor}

\usepackage{enumitem}
\usepackage{algorithm}
\usepackage[noend]{algpseudocode}
\usepackage{array}
\usepackage{multirow}
\usepackage{wrapfig}

\usepackage{listings}

\usepackage{url}

\usepackage[all]{nowidow}
\usepackage{xspace}

\newcommand{\eat}[1]{} 
\newcommand{\cut}[1]{} 

\usepackage{amsthm}

\newcommand{\nlprevision}{}

\AtBeginDocument{%
  }

\setcopyright{acmlicensed}

\copyrightyear{2025}
\acmYear{2025}
\setcopyright{cc}
\setcctype{by-sa}
\acmConference[WWW '25]{Proceedings of the ACM Web Conference 2025}{April
28-May 2, 2025}{Sydney, NSW, Australia}
\acmBooktitle{Proceedings of the ACM Web Conference 2025 (WWW '25), April
28-May 2, 2025, Sydney, NSW, Australia}
\acmDOI{10.1145/3696410.3714822}
\acmISBN{979-8-4007-1274-6/25/04}

\begin{document}

\title{Retrieval with Learned Similarities} 

\author{Bailu Ding}
\authornote{Equal contribution.}
\email{badin@microsoft.com}
\orcid{0000-0003-4138-6379}
\affiliation{%
  \institution{Microsoft Research}
  \city{Redmond}
  \state{Washington}
  \country{USA}
}

\author{Jiaqi Zhai}
\authornotemark[1]
\email{jiaqi@jiaqizhai.com}
\orcid{0009-0004-7279-3318}
\affiliation{%
  \institution{Meta}
  \city{Bellevue}
  \state{Washington}
  \country{USA}
}

\renewcommand{\shortauthors}{Bailu Ding and Jiaqi Zhai}


\input{sec_abstract}

\cut{

\ccsdesc[500]{Information systems~Probabilistic retrieval models}
\ccsdesc[300]{Information systems~Top-k retrieval in databases}
\ccsdesc[300]{Computing methodologies~Learning to rank}
\ccsdesc[100]{Information systems~Personalization}
\ccsdesc[100]{Information systems~Question answering}
}

\ccsdesc[500]{Information systems~Similarity measures}
\ccsdesc[500]{Information systems~Top-k retrieval in databases}
\ccsdesc[500]{Information systems~Learning to rank}
\ccsdesc[500]{Information systems~Probabilistic retrieval models}
\ccsdesc[300]{Information systems~Question answering}
\ccsdesc[300]{Information systems~Recommender systems}
\ccsdesc[300]{Computing methodologies~Natural language processing}
\ccsdesc[300]{Information systems~Personalization}

\keywords{Nearest Neighbor Search, Learned Similarities, Top-K Retrieval, Vector Databases, Recommendation Systems, Question Answering} 


\maketitle

\setcounter{footnote}{0}
\input{sec_introduction}
\input{sec_mol}
\input{sec_algorithm}
\input{sec_evaluation}

\input{sec_related_work}

\input{sec_conclusion}

\cut{
\section*{Limitations}

While we have performed evaluations on top of one of the largest recommendation datasets publicly available with 674,044 items, we expect potentially different algorithmic behaviors as we scale our approximate top-$K$ algorithms to handle hundreds to millions to billions of items. Additional low-level GPU kernel optimizations, one of which being an efficient grouped top-$K$ query kernel that fetches top-$K$ for multiple query embeddings (e.g., $f_1(q), \ldots, f_{P_q}(q)$), could also change the relative efficiency of approximate methods we discussed, in particular \topkavg{($N$)} and \topkcomb{$N_1$}{$N_2$}. We also leave more efficient implementations of the two-pass exact algorithm as future work.
}

\cut{
\paragraph{Future work}

While we have performed evaluations on top of one of the largest recommendation datasets publicly available with 674,044 items, we expect potentially different algorithmic behaviors as we scale our approximate top-$K$ algorithms to handle hundreds to millions to billions of items. Additional low-level GPU kernel optimizations, one of which being an efficient grouped top-$K$ query kernel that fetches top-$K$ for multiple query embeddings (e.g., $f_1(q), \ldots, f_{P_q}(q)$), could also change the relative efficiency of approximate methods we discussed, in particular \topkavg{($N$)} and \topkcomb{$N_1$}{$N_2$}. We also leave more efficient implementations of the two-pass exact algorithm as future work.
}


\bibliographystyle{ACM-Reference-Format}
\balance
\bibliography{references}


\appendix
\input{sec_appendix}

\end{document}

%% file: main_common.tex
\newcommand{\submit}{}

\ifdefined\submit
\newcommand{\todo}[1]{}
\newcommand{\changed}[1]{#1}

\newcommand{\toread}[1]{}
\newcommand{\done}[1]{}
\newcommand{\outline}[1]{}
\newcommand{\point}[1]{}
\newcommand{\todofigure}[1]{}
\newcommand{\githuburl}{\url{https://github.com/bailuding/rails}}
\long\def\tocut#1{}
\else
\ifdefined\revision
\newcommand{\todo}[1]{}
\newcommand{\changed}[1]{{\color{blue}#1}}
\newcommand{\done}[1]{}
\newcommand{\outline}[1]{}
\newcommand{\point}[1]{}
\newcommand{\toread}[1]{{\color{blue}[ToRead!] #1}}
\long\def\tocut#1{}
\else
\newcommand{\todo}[1]{\textcolor{red}{\bf [TODO!: #1]}}
\newcommand{\changed}[1]{{\color{blue}#1}}
\newcommand{\tocut}[1]{\textcolor{red}{\it\st{#1}}}
\newcommand{\done}[1]{\textcolor{blue}{\bf [DONE!] #1}}
\newcommand{\outline}[1]{\textcolor{teal}{\bf [OUTLINE!: #1]}}
\newcommand{\point}[1]{\textcolor{teal}{\bf [POINT!: #1]}}
\newcommand{\toread}[1]{{\color{blue}[ToRead!] #1}}
\newcommand{\todofigure}[1]{{\color{red}[Figure]: #1}}
\newcommand{\githuburl}{the following GitHub repository: \url{https://anonymous.4open.science/r/rails-4E62}}
\fi       
\fi

\newcommand{\mol}{\textit{MoL}\xspace}

\newtheorem{theorem}{Theorem}
\newtheorem{lemma}{Lemma}
\newtheorem{definition}{Definition}

\newcommand{\mlonem}{\textit{ML-1M}\xspace}
\newcommand{\mltwentym}{\textit{ML-20M}\xspace}
\newcommand{\amzn}{\textit{Books}\xspace}
\newcommand{\nq}{\textit{NQ320K}\xspace}

\newcommand{\topkbf}{\textit{BruteForce}}
\newcommand{\topknaive}[1]{\textit{TopKPerEmbd}#1}
\newcommand{\topkavg}[1]{\textit{TopKAvg}#1}
\newcommand{\topkcomb}[2]{\textit{CombTopK}#1\_#2}

\newcommand{\maxspeedup}[0]{\changed{\nlprevision{$66\times$}}\xspace} 
\newcommand{\nqmaxspeedup}[0]{$66\times$\xspace} 

%% file: sec_abstract.tex
\begin{abstract}

\cut{
Retrieval plays a fundamental role in recommendation systems, search, and natural language processing (NLP) by efficiently finding relevant items from a large corpus given a query. Dot products have been widely used as the similarity function in such tasks, enabled by Maximum Inner Product Search (MIPS) algorithms for efficient retrieval. However, state-of-the-art retrieval algorithms have migrated to learned similarities. These advanced approaches encompass multiple query embeddings, complex neural networks, direct item ID decoding via beam search, and hybrid solutions. Unfortunately, we lack efficient solutions for retrieval in these state-of-the-art setups. Our work addresses this gap by investigating efficient retrieval techniques with expressive learned similarity functions. We establish Mixture-of-Logits (MoL) as a universal approximator capable of expressing all learned similarity functions, and demonstrate MoL's practical applications in recommendation systems and question answering tasks. We further propose techniques to retrieve the approximate top $K$ results using MoL with a tight bound. Through extensive experimentation, we show that MoL, enhanced by our proposed mutual information-based load balancing loss, sets new state-of-the-art results across heterogeneous scenarios, including sequential retrieval models in recommendation systems and finetuning language models for question answering; and our approximate top-k retrieval with learned similarities outperforms baselines by up to \maxspeedup in latency, while achieving $>.99$ recall rate compared to exact algorithms.
}

Retrieval plays a fundamental role in recommendation systems, search, and natural language processing (NLP) by efficiently finding relevant items from a large corpus given a query. Dot products have been widely used as the similarity function in such tasks, enabled by Maximum Inner Product Search (MIPS) algorithms for efficient retrieval. However, state-of-the-art retrieval algorithms have migrated to learned similarities. These advanced approaches encompass multiple query embeddings, complex neural networks, direct item ID decoding via beam search, and hybrid solutions. Unfortunately, we lack efficient solutions for retrieval in these state-of-the-art setups. Our work addresses this gap by investigating efficient retrieval techniques with expressive learned similarity functions. We establish Mixture-of-Logits (MoL) as a universal approximator of similarity functions, demonstrate that MoL's expressiveness can be realized empirically to achieve superior performance on diverse retrieval scenarios, and propose techniques to retrieve the approximate top-$k$ results using MoL with tight \changed{error} bounds. Through extensive experimentation, we show that MoL, enhanced by our proposed mutual information-based load balancing loss, sets new state-of-the-art results across heterogeneous scenarios, including sequential retrieval models in recommendation systems and finetuning language models for question answering; and our approximate top-$k$ algorithms outperform baselines by up to \maxspeedup in latency while achieving $>.99$ recall rate compared to exact algorithms.~\footnote{Our code and model checkpoints are available at \githuburl.}

\end{abstract}

%% file: sec_introduction.tex
\vspace{.2em}
\section{Introduction}

\label{sec:introduction}

Retrieval requires efficient storing, indexing, and querying relevant candidate items represented by high-dimensional vectors. Retrieval is widely used as the initial preprocessing stage for \nlprevision{internet} applications such as recommendations, search, \nlprevision{question answering}, and natural language processing that operate over corpus with up to billions of items~\citep{ytdnn_goog_recsys16,gillick2018endtoend,mind_baba_cikm19,dpr_emnlp20,rag_neurips20,retro_dm_icml22}. \cut{Vector databases generalize the retrieval stage in a typical information retrieval pipeline, where the query- and the item- embeddings are obtained in a dual-encoder setup, and dot products are applied on top of such embeddings as the similarity function for measuring relevance.} In many concrete use cases, such as vector databases~\citep{faiss_tbd21}, the query- and the item- embeddings are learned with deep neural networks in a dual-encoder setup, and dot products are applied on top of such embeddings as the similarity function for measuring relevance.

Despite the popularity of dot products and numerous work done to improve their efficiency~\citep{pq_nns_pami11,alsh_ping_neurips2014,hnsw_tpami18,tpuknn_goog_neurips22}, state-of-the-art retrieval algorithms have long moved to various learned similarity functions. Their most basic versions preserve some dot product-related structures, but turn either the query or the item into multiple embeddings, and rely on a max operator to combine those similarity values~\citep{mind_baba_cikm19,colbert_sigir20}. As another example, Probabilistic Label Trees (PLTs)~\citep{plt_icml16} and Tree-based Deep Models (TDMs)~\citep{tdm_kdd18,otm_icml20} map items to leaf nodes in a tree, and reduce retrieval to beam search by making decisions sequentially using learned classifiers while traversing trees from root to leaf. More recent work on generative retrieval 
directly map the query to the item ids in sequence-to-sequence or decoder-only setups~\citep{genre_iclr21,seal_neurips22,dsi_neurips22,nci_neurips22,genret_neurips23}. Combinations of these approaches have also been studied, with some performing coarse-grained retrieval with generative approaches, followed by re-ranking using dot products~\citep{dr_cikm21}. Finally, the similarity function can be directly parameterized by carefully designed deep neural networks that take various forms~\citep{ncf_www17,mos_iclr18,colbertv2_naacl22,zhai23kdd}.

Supporting efficient retrieval with these diverse learned similarities\cut{in vector databases} is challenging. Learned similarity functions are generally expensive to compute; with learned index structures, traversing a binary tree with 4 million items requires running beam search for 20 non-parallelizable steps~\citep{tdm_kdd18}, while 
recommendation and NLP deployments commonly need to handle billions of items~\citep{pixie_pins_www18,mind_baba_cikm19,borisyuk2024linr_cikm24} with \nlprevision{a latency budget of tens of milliseconds}.
When an arbitrary deep neural network is employed, it's no longer clear how to perform top-$K$ retrieval other than through brute-force~\citep{ncf_www17} or heuristics~\citep{zhai23kdd}. \nlprevision{While graph-based methods can be used to prune the search space~\citep{hnsw_tpami18,diskann_neurips19,flashlight_log22,cagra_icde24}, such methods tend to be much slower compared with MIPS algorithms leveraging quantization at high recall rates~\citep{scann_icml20,annbenchmarks-website}, and their performance can degrade when the similarity function is not a distance metric~\citep{yandex2018neurips}.}
What is worse, these algorithms vary significantly in terms of their exact formulations, and the lack of a universal interface makes it even more difficult \nlprevision{to design} a general solution for efficient retrieval.

Taking a step back, our key insight is that learned similarity approaches are but different ways to increase \nlprevision{the} expressiveness of the retrieval stage. Formally, for a query $q$ and an item $x$, the expressiveness of the similarity function boils down to deriving alternative parameterizations of $p(x | q)$ matrices, with full rank matrices being the most expressive among them. Dot products, on the other hand, induces a low-rank bottleneck due to the dimensionality of the embedding, i.e.,  $\ln p(x | q) \propto \langle f(q), g(x)\rangle$ ($f(q), g(x) \in \mathbb{R}^d$). This cannot be alleviated by simply increasing the embedding dimension $d$, due to memory bandwidth being the main bottleneck in modern dot-product based retrieval systems, such as vector databases~\citep{faiss_tbd21,tpuknn_goog_neurips22,zhai23kdd}, and overfitting issues 
that come with larger embedding dimensions due to the common need to co-train or finetune query- and item- encoders from data~\citep{ytdnn_goog_recsys16,mind_baba_cikm19,dpr_emnlp20, dr_cikm21,sentencet5_acl2022,gtr_emnlp22,zhai2024actions_icml24}. 

This insight enables us to support efficient retrieval with \nlprevision{expressive} learned similarity functions by approximating them with Mixture-of-Logits (MoL).
To the best of our knowledge, this is the first work that tackles the problem of efficient retrieval with universal learned similarities, \nlprevision{while setting new state-of-the-art results across \emph{heterogeneous} scenarios}.
We first show that Mixture-of-Logits  is a universal approximator as it can express $p(x|q)$ matrices of arbitrary high rank, and hence approximate \emph{all} learned similarity functions \nlprevision{(Section~\ref{sec:mol_expressiveness})}.  \nlprevision{
Our work lays theoretical foundations for MoL's empirical impressive performance gains of 20\%-30\% across Hit Rate@50-400 on web-scale corpus with hundreds of millions to billions of items~\citep{zhai23kdd,borisyuk2024linr_cikm24}, and further enables MoL to be effectively applied across diverse retrieval scenarios, from large-scale recommendation systems to finetuning language models for question answering (Section~\ref{sec:mol-adaptation})}.
We next propose techniques to retrieve the approximate top-$K$ results using MoL with a tight \changed{error} bound (Section~\ref{sec:algorithm}).
Our solution leverages the existing widely used APIs of vector databases like top-K queries, thus benefiting from prior work on efficient vector search like MIPS~\citep{pq_nns_pami11,alsh_ping_neurips2014,faiss_tbd21,scann_icml20}.
We empirically compare our techniques with existing approaches, showing that MoL sets new state-of-the-art results on recommendation retrieval \nlprevision{and question answering} tasks, and our approximate top-k retrieval with learned similarities outperforms baselines by up to \maxspeedup in latency, while achieving $>.99$ recall rate of exact algorithms (Section~\ref{sec:evaluation}). Importantly, our approach with learned similarities efficiently utilizes modern accelerators due to MoL's higher arithmetic intensity~\citep{zhai23kdd}, which results in MIPS-level inference latency and throughput. Overall, our work provides strong theoretical and practical justifications to migrate away from the broadly adopted MIPS solution in vector databases to Retriev\underline{a}l w\underline{i}th Learned Similarities (RAILS) on GPUs.

%% file: sec_mol.tex
\vspace{-.8em}
\section{Mixture of Logits}
\label{sec:mol}
\input{table_notations}

In this section, we describe Mixture of Logits (MoL), \nlprevision{propose a load balancing loss to improve conditional computations in MoL}, prove that MoL is expressive enough to represent any learned similarity function, and demonstrate how to apply MoL to \changed{diverse} retrieval tasks. Table~\ref{table:notations} summarizes the notations in this paper.

We first describe Mixture of Logits (MoL).

\paragraph{Mixture of Logits (MoL)} MoL~\citep{zhai23kdd} assumes that the query $q$ and the item $x$ are already mapped to $P$ \changed{pairs} of low-rank embeddings (``component-level embeddings''), $f_{p}(q), g_p(x) \in \mathbb{R}^{d_P}$, where $f_{p}(q), g_p(x)$ are parameterized with some neural networks based on query and item features, respectively, and $d_P$ is the dimensionality of the low-rank embeddings. MoL then calculates the similarity between the query $q$ and the item $x$ by applying adaptive gating weights, $\pi_p(q, x) \in [0, 1]$, to the inner products of these $P$ pairs of low-rank embeddings, or $\langle f_{p}(q), g_p(x) \rangle$s. Note that prior work assumes that $\sum_p \pi_p(q, x) = 1$~\citep{zhai23kdd,borisyuk2024linr_cikm24}, 
but this does not affect our analyses in this paper.
Following~\cite{zhai23kdd}:
\vspace{-.3em}
\begin{equation}
\phi(q, x) = \sum_{p=1}^P \pi_p(q, x) \langle f_p(q), g_p(x)\rangle
\label{eq:mol-naive}
\end{equation}
\vspace{-.8em}

To extend this to large-scale datasets and to enable hardware-efficient implementations on accelerators like GPUs, Equation~\ref{eq:mol-naive} was further modified by decomposing those $P$ dot products as (batched) outer products of $P_q$ query-side and $P_x$ item-side embeddings, where $P_q \times P_x = P$, and applying l2-norm to \changed{the embeddings}: 

\begin{equation}
\phi(q, x) = \sum_{p_q=1}^{P_q} \sum_{p_x=1}^{P_x} \pi_{p_q, p_x}(q, x) \left\langle \frac{f_{p_q}(q)}{||f_{p_q}(q)||_2}, \frac{g_{p_x}(x)}{||g_{p_x}(x)||_2}\right\rangle 
\label{eq:mol-simplified}
\end{equation}
We use Equation~\ref{eq:mol-naive} and~\ref{eq:mol-simplified} interchangeably as the MoL form to analyze throughout the rest of this paper, given that the embedding normalization for $f_{p_q}(q)$s and $g_{p_x}(x)$s can be precomputed.

\paragraph{Mixture of Logits (MoL) with load balancing regularization loss.} \nlprevision{We further observe $\pi_p(q, x)$ defines conditional computation to be performed over the $p$ low-rank embedding pairs, or $(f_p(q), g_p(x))$s. $\pi_p(q, x)$ should hence satisfy two conditions: 
\begin{itemize}[leftmargin=*]
\item Globally, the $p$ low-rank embedding pairs, or $(f_p(q), g_p(x))$s, should receive a similar number of training examples even when $p$ is large and $\pi_p(q, x)$ is sparse, with load distributed evenly across the $p$ pairs. One way to do this is to maximize the entropy $H(p)$ over these embedding pairs.
\item The low-rank embedding pairs used to compute \changed{a particular} $\phi(q, x)$ should be non-uniform and ideally sparse; e.g., it's desirable to avoid the degenerate solution where $\pi_p(q, x) = \frac{1}{p}$. One way to do this is to minimize the conditional entropy $H(p|(q, x))$ of $p$ given (query, item) pairs.
\end{itemize}
Given these two desired conditions, we propose a mutual information-based regularization loss for load balancing, defined as 
\begin{equation}
\mathcal{L}_{MI} = - H(p) + H(p|(q, x))
\label{eq:mi-regularization-loss}
\end{equation}
with the overall training loss as 
\begin{equation}
-\log \frac{\exp(\phi(q, x))}{\exp(\phi(q, x)) + \sum_{x' \in \mathbb{X}} \exp(\phi(q, x'))} + \alpha \mathcal{L}_{MI}
\label{eq:mol-loss}
\end{equation}
where the first part of Equation~\ref{eq:mol-loss} is the sampled softmax loss used in~\cite{zhai23kdd}, and the second part adjusts the weight for the mutual information-based load balancing loss with a hyperparameter $\alpha$. 
}

\input{figure_mol}

\subsection{Expressiveness of Mixture of Logits}
\label{sec:mol_expressiveness}
Now we show that any high-rank matrix can be decomposed into a mixture of logits based on low-rank matrices, i.e., MoL is a universal approximator \changed{for all similarity functions}. Without loss of generality, we prove the following:

\begin{theorem}
    \label{theorem:main}
    \textbf{MoL decomposition}: Let $A$ be a matrix of $n\times m$, where $n\leq m$. There exists $\pi_1, B_1, \pi_2, B_2, \cdots, \pi_p, B_p$ such that $|A-\sum_{p=1}^{P}\pi_p \circ B_i|<\epsilon$, where $\epsilon$ is a small positive number. Here $B_i$ is a matrix of $n\times m$ with rank equal to or less than $d$, and $\pi_1, \pi_2, \cdots, \pi_P$ are $n\times m$ matrices that together define a probability distribution over each $(i, j)$ tuple, such that $\sum_{p=1}^{P}\pi_{p}(i,j)=1, 0\leq \pi_{p}(i,j) \leq 1$ for any $1\leq i\leq n, 1\leq j\leq m, 1\leq p\leq P$.
\end{theorem}

We can think about $n$ as the number of queries and $m$ the number of items (or vice versa). First, the theorem trivially holds if the rank of $A$ is less than or equal to $d$ ($d \leq n$):

\begin{lemma}
    \textbf{MoL decomposition when $Rank(A)\leq d$}: Let $A$ be a matrix as defined in Theorem~\ref{theorem:main}. If the rank of $A$ is less than or equal to $d$, then we have $A = \pi \circ A$, where $\pi(i,j) =1$ for any $1\leq i\leq n, 1\leq j\leq m$.
\end{lemma}

\cut{
\todo{Fix the proof to prove $Rank(A)>d$ instead of $Rank(A)=n$}
\todo{Fix the proof to cover the case where $\sum_{k=1}^dU_{ik}V_{kj}=0$. Need to approximate the matrix with $|A-A'|<\epsilon$}
}

Then we prove for the case where the rank of $A$ is greater than $d$. Without loss of generality, we prove the case where the matrix has full rank, i.e., $Rank(A)=n$:

\begin{lemma}
    \textbf{MoL decomposition when $Rank(A)=n$}: Let $A$ be a matrix as defined in Theorem~\ref{theorem:main}. Then there exists $\pi, B_1, B_2$ such that $|A-(\pi \circ B_1 + (1-\pi) \circ B_2)|< \epsilon$, where $Rank(B_1)\leq d$, $Rank(B_2)\leq d$, and $0\leq \pi(i,j) \leq 1$ for $1\leq i\leq n, 1\leq j \leq m$. 
\end{lemma}

\begin{proof}
    Because $A$ is a matrix of rank $n$, it can be rewritten as $A=UI_nV$, where $I_n$ is an identity matrix with rank $n$. \cut{\textcolor{red}{jz: might be good to call out $m$ $n$ explicitly in the following writing}} 
    Thus, $A_{ij} = \sum_{k=1}^n U_{ik}V_{kj}, 1\leq i\leq n, 1\leq j\leq m$. Let $A'$ be a matrix of $n\times m$, where $A'_{ij} = \lambda_{ij} \cdot \sum_{k=1}^d U_{ik}V_{kj}$ for $1\leq i\leq n, 1\leq j\leq m$. Here, $\lambda_{ij} =1 + \frac{\sum_{k=d+1}^nU_{ik}V_{kj}}{\sum_{k=1}^d U_{ik}V_{kj}}$ if $\sum_{k=1}^d U_{ik}V_{kj} \neq 0$, otherwise $\lambda_{ij} = 1 + \frac{\sum_{k=d+1}^nU_{ik}V_{kj}}{\epsilon}$. Thus, we have $|A-A'|\leq \epsilon$.
    
    Let $\lambda_{min} = \min{\lambda_{ij}}$, and $\lambda_{max}=\max{\lambda_{ij}}$. Let $B_1 = \lambda_{min}U D_{n,d} V$, $B_2 = \lambda_{max}U D_{n,d} V$, where $D_{n,d}$ denotes an $n$-by-$n$ diagonal matrix with the first $d$ elements of the diagonal being $1$s and the rest being $0$s. We have $A'_{ij} =\lambda_{ij} \sum_{k=1}^d U_{ik}V_{kj} = \pi(i,j)\cdot B_{1ij} + (1-\pi(i, j))\cdot B_{2ij}$, where $\pi(i,j) = \frac{\lambda_{max} - \lambda_{ij}}{\lambda_{max} - \lambda_{min}}$. Because $\lambda_{min}\leq \lambda_{ij}\leq \lambda_{max}$, we have $0\leq \pi(i, j) \leq 1$.     
    \cut{\textcolor{red}{jz: may want to use something other than $\pi_{i,j}$ here}}
    
    Thus, we have constructed $B_1, B_2, \pi$ such that $|A-(\pi\circ B_1+(1-\pi) \circ B_2)|=|A-A'|\leq \epsilon$.
\end{proof}

\textit{Remark}
    Here, we have shown that any high-rank matrix can be expressed as a mixture of logits of two low-rank matrices. Note that our decomposition is not intended to be used as a distillation of the original high-rank matrix. 
    It is likely prohibitively expensive to populate the full matrix with a learned similarity function. In addition, our proof also does not indicate that having
    two mixture components
    is sufficient to train the embeddings and the learned similarity function. It is well-known that overparameterization is often necessary to enable efficient and performant training. 

\subsection{Applying MoL to Heterogeneous Use Cases}
\label{sec:mol-adaptation}

\nlprevision{We now discuss how to apply MoL to retrieval tasks in different domains. Parameterization of the low-rank, component-level embeddings, or $f_{p}(q), g_p(x) \in \mathbb{R}^{d_P}$, plays an important role in realizing MoL's theoretical expressiveness in practice, as suggested by prior work~\cite{borisyuk2024linr_cikm24}. We discuss two scenarios on the opposite end of the spectrum, one with \emph{a large number of heterogeneous features} -- retrieval in large-scale recommendation systems, followed by another with \textit{a single homogeneous feature} -- finetuning language models for question answering and related NLP use cases, shown in Figure~\ref{fig:mol-adaptation}.}

\paragraph{Retrieval in Large-scale Recommendation Systems.} \nlprevision{Recommendation systems are characterized by the large number of heterogeneous features they use~\cite{ytdnn_goog_recsys16,autoint_cikm19,zhai2024actions_icml24}. This naturally enables some of those features to be utilized on the query- (user-) or on the item-side. For instance, embeddings can be constructed based on cluster ids on both the query-side and the item-side~\cite{borisyuk2024linr_cikm24}. For common benchmark datasets, User ID-based one-hot embeddings~\cite{matrix-factorization-netflix09} represent another possible $g_p(q)$ to use, which we evaluate in Section~\ref{sec:evaluation}.}

\paragraph{Finetuning Language Models for Question Answering.} \nlprevision{In contrast, language models are characterized by their use of homogeneous semantic features, such as wordpieces and sentencepieces~\cite{sentencepiece_emnlp18}. We observe that MoL can be similarly adopted for those use cases. To obtain the $P_X$ item embeddings for MoL, we expand \changed{the} tokenizer's vocabulary with $P_X$ \emph{special aggregation tokens} $X_1, \ldots, X_{P_X}$, and append those $P_X$ tokens at the beginning of every tokenized sequence, $SP_1, \ldots, SP_N$, as illustrated in Figure~\ref{fig:mol-adaptation}~\footnote{Note that many question answering scenarios~\cite{dpr_emnlp20,genre_iclr21,nci_neurips22,gtr_emnlp22,genret_neurips23} utilize bidirectional language models for retrieval, like BERT~\cite{bert_naacl19} or T5~\cite{t5_raffel2023exploringlimitstransferlearning}; for  recent unidirectional language models, we can add $X_1, \ldots, X_{P_X}$ to the end of the input sequence instead.}. These $P_X$ special tokens play similar roles as the CLS token in BERT~\cite{bert_naacl19}, and during finetuning of the language model, are co-trained to aggregate different aspects of information as inputs for MoL.
Additionally, we can design a learned pooling function to adapt pooling policy at an example-level (``Parameterized Pooling'') to improve model quality, which we discuss further in Appendix~\ref{sec:app-exp-emb-parameterization-qa}.
} 


\input{figure_mol_adaptation}

%% file: table_notations.tex
\begin{table*}[t]
\vspace{-.5em}
\begin{center}
\small
\begin{tabular}{c|l}
\toprule
\bf Notation                  & \bf Description \\
\hline
$q$ ($Q$, $|Q|$) & query (set of queries, number of queries) \\
$x$ ($X$, $|X|$) & item (set of items, number of items) \\
$\phi(q, x)$ & the learned similarity function, i.e., Mixture-of-Logits (MoL). \\
\multirow{2}{*}{$P$ ($P_q$, $P_x$)} &\multirow{2}{13cm}{MoL uses $P$ \changed{pairs of} low-rank embeddings ("component-level embeddings") to represent $q$ and $x$. With the (batched) outer product form of MoL, $P_q$ and $P_x$ are the numbers of embeddings for $q$ and $x$, respectively; $P = P_q \times P_x$.} \\ 
& \\
$\pi_{p}(q,x)$ ($\pi_{p_q,p_x}(q,x)$) & weight for the $p$-th (or $p_q$-th by $p_x$-th with outer product) embedding set for $(q, x)$. \\
$f(q)$ ($f_p(q)$) & learned embedding for the query ($p$-th component-level query embedding) \\
$g(x)$ ($g_p(x)$) & learned embedding for the item  ($p$-th component-level item embedding) \\
$d_P$             & dimensionality of low-rank (component-level) embeddings. $f_p(q), g_p(q) \in \mathbb{R}^{d_P}$. \\
$\langle f(q), g(x)\rangle$ & the dot product similarity function: $g(x)^T f(q)$. \changed{$\langle f_p(q), g_p(x)\rangle$ denotes the dot product for the $p^{th}$ embedding pair.}\\
\bottomrule
\end{tabular}
\vspace{0.5em}
\caption{Table of Notations.}
\label{table:notations}
\ifdefined\singlecolumnpdf
\vspace{0.0em}
\else
\vspace{-2em}
\fi
\end{center}
\end{table*}

%% file: figure_mol.tex
\begin{figure*}[t]
    \vspace{-0.9em}
    \centering
    \includegraphics[width=0.75\linewidth]{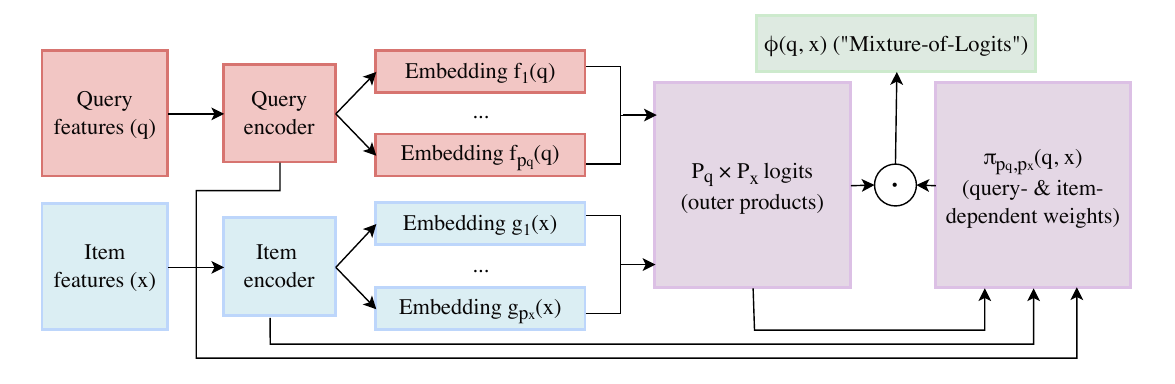}
    \vspace{-1.5em}
    \caption{Mixture-of-logits (MoL) learned similarity.}
    \label{fig:mol}
    \vspace{-0.6em}
\end{figure*}

%% file: figure_mol_adaptation.tex
\begin{figure*}[t]
    \vspace{-1em}
    \centering
    \includegraphics[width=0.7\linewidth]{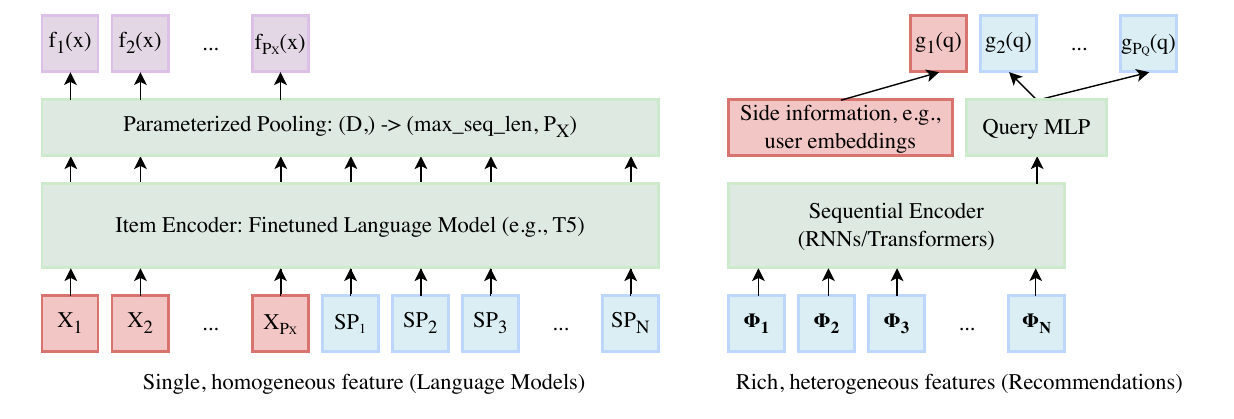}
    \vspace{-1.2em}
    \caption{Illustration of how to \nlprevision{apply} Mixture-of-logits (MoL) learned similarity to various retrieval scenarios, with a language model finetuning use case (characterized by a single homogeneous feature) shown on the left, and a recommendation use case (characterized by a large number of heterogeneous features) shown on the right. More details can be found in Appendix~\ref{sec:app-exp-emb-parameterization-qa}.}
    \label{fig:mol-adaptation}
    \vspace{-1em}
\end{figure*}

%% file: sec_algorithm.tex
\section{Retrieval Algorithms}
\label{sec:algorithm}

In this section, we describe the problem of retrieving the top $K$ items with \mol as well as exact and approximate retrieval algorithms. \cut{to retrieve the top $K$ documents with the highest similarity scores given a query $q$. 
Given that the learned similarity function is expensive, our goal is to invoke the learned similarity functions on a small number of items. At high level, we propose a two-stage retrieval algorithm, where we first leverage cheap dot-product based similarity functions to approximately retrieve candidate items among all the items, and then we invoke the similarity function on the candidate items to find the top $K$ items.} Formally, we define the top $K$ retrieval problem as \changed{follows}: 

\begin{definition}
    \textbf{Top $K$ with \mol}: Let $q$ be a query and $X$ be a set of items, where both the query $q$ and each item $x \in X$ are associated with $P$ embeddings. Together we have $P$ pairs of embeddings, $(f_p(q), g_p(x)), 1\leq p\leq P$. Let $\phi(q, x)=\sum_{p=1}^P\pi_{p}(q,x) \langle f_p(q), g_p(x)\rangle$ be the similarity score of $q, x$, where $x\in X$. 
    The top $K$ query with \mol returns the $K$ items from $X$ with the highest $\phi(q,x)$s.
\end{definition}

For approximate top $K$ retrieval with \mol, we define the gap of the approximate and exact top $K$ results as follows:
\begin{definition}
    \textbf{Gap of approximate top $K$:} Let $q$ be a query and $X_K$ be the set of exact top $K$ items for the query $q$ from a set of items $X$. Let $X^*$ be the \changed{set of} approximate top $K$ results, where $X^*\subseteq X$. Let $S=\min\{\phi(q, x), x\in X^*\}$ and $S'=\max\{\phi(q, x), x\in X_K\setminus X^*\}$. We call $S_{\Delta}=S'-S$ the \emph{gap} of the top $K$ with $X^*$.
\end{definition}

\subsection{Exact algorithm}
\label{sec:algorithm:exact}

The brute-force algorithm to retrieve the exact top $K$ with \mol is to evaluate $\phi(q, x)$ for each query $q$ and item $x$. This algorithm can be prohibitively expensive if the number of items is large. Instead, we describe a more efficient two-pass algorithm to retrieve the exact top $K$ items as shown in Algorithm~\ref{code:algorithm:exact} \changed{(example in Appendix~\ref{appendix:example})}.

\input{code_algorithm_exact}

We start by retrieving the top $K$ items with the highest dot product scores for each group of embeddings as the initial candidate set $G$ (line 1-4). Then we evaluate the \mol scores of the items in $G$ and find the minimal MoL score $S_{min}$ (line 5-8). Next we retrieve all items within a distance of $S_{min}$ with the query $q$ as the candidate set $G'$ (line 9-11). Finally, we evaluate the \mol scores of the items in $G'$, 
and return the top $K$ items with the highest scores (line 12).

We argue that Algorithm~\ref{code:algorithm:exact} retrieves the exact top $K$ items with \mol. Let $X_K$ be the set of the exact top $K$ items and $X'$ be the result of Algorithm~\ref{code:algorithm:exact}. Let $x\in X_K$ and $\phi(q, x)$ be the \mol score of $x$ and $q$. Since $x$ has the highest top $K$ score with \mol, $\phi(q,x)\geq S_{min}$. Since the \mol score is a weighted score over the dot product scores, we have $\max\{\langle f_p(q), g_p(x)\rangle, 1\leq p\leq P\}\geq \phi(q, x)\geq S_{min}$. Since Algorithm~\ref{code:algorithm:exact} retrieves all the items with a dot product higher than or equal to $S_{min}$ of $q$ for each embedding $q_p$ (line 9-11), we have $x\in G'$. Thus, $x\in X'$. So we have shown that $X_K=X'$. 
\cut{
\changed{
\paragraph{Example}
\input{table_example}
Table~\ref{table:example} shows an example of $4$ items with $2$ embedding sets. Assume the goal is to retrieve the top $2$ items. In the first stage, we retreive the top $2$ items for each embedding set, i.e., item $a, b, c$ are retrieved. For each of the retrieved items, we calculate the MoL scores based on their gating weights, i.e., $1.0, 0.4, 0.4$ for $a, b, c$ respectively. Here, $S_{min}$ is set to $0.4$. In the second stage, we retrieve all the items with $\langle f, g\rangle \geq 0.4$ for each embedding set, i.e., item $d$, and then calculate their corresponding MoL score, i.e., $0.7$. The algorithm returns $a$ and $d$ as the top $2$ items.
}
}

\vspace{-.4em}
\subsection{Approximate algorithms}
\label{sec:algorithm:approximate}
\changed{In the exact algorithm shown in Algorithm~\ref{code:algorithm:exact}, we need to retrieve all the items with a dot product higher than or equal to a threshold. When the threshold is a loose filter of the item set, which may happen when the dot products are skewed, $G'$ can be large, and the evaluation of \mol over a large number of candidates can be expensive. Here, we describe two heuristics to approximately retrieve the top $K$ items and analyze their gap against the exact algorithm.}

In both heuristics, we perform a two-stage retrieval as shown in Algorithm~\ref{code:algorithm:approximate}. In the first stage, we retrieve a set of $K'$ candidate items that are potentially high in \mol score by using dot products (line 2). Note that $K'$ can be larger than $K$, e.g., due to oversampling. In the second stage, we evaluate the \mol scores of the candidate items and return the top $K$ items (line 3).

\input{code_algorithm_approximate}

Here, we describe two heuristics to retrieve the candidate items: 

\paragraph{Top $K$ per embedding.} Given a query $q$ and a set of items $X$, for each \changed{of the $p$ embedding pairs}, retrieve top $K$ items $X_{K,p}$ based on dot product ($\langle f_p(q), g_p(x)\rangle$). Return the union across $P$ queries. 

The top $K$ per embedding heuristic returns the union of the top $K$ items for each embedding \changed{pair (group)} by dot product. We analyze the gap of this \changed{approach} as follows:
\begin{theorem}
\label{theorem:topkgap}
    \textbf{Upper bound of the gap of top $K$ per embedding:} Let $X_{K,p}$ be the top $K$ items of the embedding set $p$ and \changed{$S=\max\{\langle f_p(q), g_p(x) \rangle, x\in X_{K+1,p}\setminus X_{K, p}, \forall p\}$}. Let \changed{$S_K$} be the \nlprevision{$K^{th}$ largest} \mol score of the items in $\cup_p X_{K, p}$, then we have 
    \changed{$S_{\Delta}\leq S-S_K$}.
\end{theorem}

\textit{Remark}
    This gap (\changed{error bound}) bounds the maximal difference between the \changed{$K^{th}$ largest MoL score from the set of items} retrieved by the heuristic and the actual \changed{$K^{th}$ largest MoL score}. \changed{In addition, any retrieved item $x$ with $\phi(q, x) \geq S$ belongs to the actual top $K$ items.} Note that there exists an \mol such that \changed{$S_{\Delta}=S-S_K$}, i.e., when $\pi_p(q, x)=1$ for \sloppy\changed{$x, p=\argmax_{x, p}\{\langle f_p(q), g_p(x)\rangle, x\in X_{K+1,p}\setminus X_{K, p}, \forall p\}$}. Thus, the upper bound of $S_{\Delta}$ is tight. \changed{In practice, we can calculate a looser bound with the items retrieved by per embedding top $K$, i.e., from $X_{K, p}$. We provide an example in Appendix~\ref{appendix:example}.}

\cut{
\changed{
\paragraph{Example}
    Consider the example shown in Table~\ref{table:example}. Assume we want to retrieve the top $2$ items. With per component embeddings, item $a, b, c$ are retrieved, and $S_k=0.4$. Since $S$ is calculated as $\max\{0.7, 0.2\}$, the upper bound of the gap is $0.3$. Here, the gap is exact, i.e., the MoL score of the actual top $2^{nd}$ item $d$ is $0.3$ higher than that of the top $2^{nd}$ item from the retrieved items $\{a, b, c\}$. If we bound the gap with the retrieved items only, i.e., $a, b, c$, then we will get a looser bound of $0.8-0.4=0.4$.
}
}

\paragraph{Top $K$ average.} Given a query $q$ and a set of items $X$, return the top $K$ items with the highest average dot product \changed{defined as} $\sum_{p}\langle f_p(q), g_p(x) \rangle/P$.

Note that the top $K$ average heuristic returns the exact top $K$ items when the gating weight distribution in \mol, $\pi$, is uniform. 

This heuristic is interesting for two reasons. First, the items retrieved by this heuristic are likely to be the top $K$ items of \mol when the weight distribution is more balanced. This complements the heuristic that retrieves top $K$ per embedding. Second, in the setup where the set of embedding pairs is constructed as the outer product of the embeddings of a query and those of an item (Equation~\ref{eq:mol-simplified}), the average dot product can be efficiently preprocessed and materialized for the items, and the computation of the top $K$ average is then \emph{agnostic} to the number of embedding pairs. 

Formally, let $P=P_q\times P_x$ be the number of embedding pairs, where $P_q$ is the number of embeddings of a query $q$ and $P_x$ is that of an item $x$. The average dot product can be computed as

\ifdefined\singlecolumnpdf
\begin{align}
     \frac{1}{P}\cdot \sum_{p=1}^P{\langle f_p(q), g_p(x)\rangle} = \frac{1}{P} \cdot \sum_{p_q=1}^{P_q}\sum_{p_x=1}^{P_x}{\langle f_{p_q}(q), g_{p_x}(x)\rangle} = \frac{1}{P} \cdot \left\langle \sum_{p_q=1}^{P_q} f_{p_q}(q), \sum_{p_x=1}^{P_x} g_{p_x}(x)\right\rangle
\end{align}
\else
\vspace{-.5em}
\begin{align}
     \frac{1}{P}\cdot \sum_{p=1}^P{\langle f_p(q), g_p(x)\rangle} &=& \frac{1}{P} \cdot \sum_{p_q=1}^{P_q}\sum_{p_x=1}^{P_x}{\langle f_{p_q}(q), g_{p_x}(x)\rangle} \\
                                                                  &=& \frac{1}{P} \cdot \left\langle \sum_{p_q=1}^{P_q} f_{p_q}(q), \sum_{p_x=1}^{P_x} g_{p_x}(x)\right\rangle
\end{align}
\vspace{-.5em}
\fi

Thus, we can preprocess the embeddings of the items and the query, so the number of embeddings accessed is $1$ per item for a given query, regardless of the overall number of component-level embeddings used by \mol, i.e., $P$. 

Finally, we can combine the candidates retrieved from top $K$ per embedding group and the top $K$ average as the following:
\paragraph{Combined top $K$.} Given a query $q$, a set of items $X$, and $K$, return the union of the items from the top $K$ per embedding group across the $P$ groups and the top $K$ items from the top $K$ average.

\begin{theorem}
    \textbf{Upper bound of the gap of combined top $K$.} Let $X_{K,p}$ be the top $K$ items of the embedding set $p$ and \changed{$S_K$} as defined in Theorem~\ref{theorem:topkgap}. Let $X'_K$ be the top $K$ items from top $K$ average. Let \changed{$S'=\max\{\langle f_p(q), g_p(x)\rangle, x\in X\setminus (X_{K, p}\cup X'_K), \forall p\}$}. Then the gap \changed{$S_{\Delta}$ satisfies $S_{\Delta}\leq S'-S_K$}.
\end{theorem}
\textit{Remark} Similar to Theorem~\ref{theorem:topkgap}, the upper bound of this gap is tight. In practice, we can configure the $K$ to be different for the two heuristics, i.e., $K_1$ and $K_2$. For example, when the weight distribution $\pi$ is more balanced, $K_2$ can be configured to be larger as the top $K$ average approach approximates \mol well while being more computationally efficient.

\cut{
\label{sec:average_approx}
In empirical evaluation, the average of the dot product scores across all embedding sets is a reasonably good approximation of the learned similarity scores with MoL. This is because the gating weights are not very skewed, i.e., max gating weights is around $0.2 - 0.3$ and the second max gating weights is around $0.1 - 0.2$ for a set of $64$ embeddings. \textcolor{blue}{note that say 50\% of the 64 embeddings should still have gating weights of 0?}

\todo{Analyze how close the average score vs.the learned similarity score under different data distributions}

We can use a two-stage algorithm to scope the items that are potentially top $K$ items with learned similarity score with the average similarity scores:
\begin{itemize}
    \item Use dot product score to find the top $K^*$ items with the highest average dot product scores across all embeddings, where $K^* > K$. For example, $K^*$ can be $10\times K, 20\times K, 100\times K$, etc.
    \item Compute the learned similarity score of the top $K^*$ items and return the top $K$ items with the highest learned similarity score.
\end{itemize}
}

\cut{
\paragraph{Cost analysis and optimization}
With the general MoL problem specification, we need to load all the embedding sets of an item and the query to compute the average of dot product scores. Thus, the number of embeddings accessed is $P$. 

In the specific problem setup where the embedding set is constructed as the cross product of the query and item embeddings, we can preprocess the item embeddings and calculate the average dot product scores as follows:

Let $P=P_u\cdot P_x$ be the number of embedding sets, where $P_u$ is the number of embeddings for a user and $P_x$ is the number of embeddings of an item. The average dot product score is computed as
\begin{align}
    \frac{1}{P}\cdot \sum_{p=1}^P{u_p\cdot x_p} &= \frac{1}{P} \cdot \sum_{r=1}^{P_u}\sum_{s=1}^{P_x}{u_r\cdot x_s} 
    &= \frac{1}{P} \cdot \sum_{r=1}^{P_u} u_r\sum_{s=1}^{P_x}x_s 
    &= \frac{1}{P} \cdot u^* \cdot x^*
\end{align}

Here, $u^*$ is the sum of all the embeddings of the query and $x^*$ is the sum of all the embeddings of the item. Thus, we can preprocess the embeddings of the items and the query, so the number of embeddings accessed is $2$ per user and item pair, regardless of the number of embedding sets. 
}
\cut{
\subsection{Cover both skewed and non-skewed gating weights}
The basic idea is to comine both the top $K$ when the gating weights are very skewed and these when the gating weights are not skewed.
\begin{itemize}
    \item When gating weights are very skewed, the top $K$ dot product of each embedding set should cover the top $K$ items with learned similarity.
    \item When the gating weights are not very skewed, the average of each embedding set should cover the top $K$ items with learned similarity.
    \item If we can combine these cases together, we should have a good coverage of the top $K$ items regardless of the learned similarity.
    \item It is possible that we can find a balanced point to decide how many candidate items to get from the average dot product score and how many candidate items to get from the top $K$ dot product score so that the total cost is minimal to guarantee a given coverage of the degree of skewness of the gating weights.
\end{itemize}
}

\cut{
\subsection{Max approximation of the MoL}
\label{sec:max_approx}
In the empirical evaluation, we also found that taking the top $K$ items whose maximal dot product across the embedding sets are the highest is not only a good approximator by itself, it also complements the average approximation. Intuitively, the average approximation provides good top $K$ candidates for more balanced MoL gating weights, while the max approximation provides good top $K$ candidates for skewed MoL gating weights.

The max approximation of MoL works as the following:
\begin{itemize}
    \item For each item, calculate the maximal dot product across all the embeddings with the query. Thus, each item now associates with a single score, i.e., the \emph{max score}.
    \item Find the top $K$ item with the highest max score.
\end{itemize}

Again, we can over retrieve for max approximation, where the number of items retrieved is $10\times K, 20\times K$, etc.
}

\cut{
\subsection{Bound analysis}
\label{sec:bound_analysis}
Before describing the approximate algorithm, we first show the bound of inaccuracy of the remaining items and the retrieved items.

Let $I_p$ be the set of items retrieved by the top $K$ dot similarity scores for an embedding set $p$, $1\leq p\leq P$. Let $I=\cup_p I_p, 1\leq p\leq P$

\paragraph{Bound of inaccuracy with remaining items}

Given a set of candidate items retrieved by top $K$ of the dot similarity score:
\begin{itemize}
    \item Input: query $q$, a set of items $G$, $\pi$, $P$ sets of embeddings $q_p, x_p, 1\leq p\leq P$.
    \item Output: a bound $B$ for the inaccuracy of remaining items $x \notin G$. Let $S^K(q, Z)$ be the $K^{th}$ highest learned similarity scores for $x\in Z$. We have $S^K(q, G)-S^K(q, X)\leq B, X\subseteq G$.
    \item For each set of embeddings $p$, let $F_p(q)$ be the $K^{th}$ highest dot similarity score for user $u$. 
    \item Thus, the potential highest learned similarity score of the item $x\notin G$ will be $F^*(q) = \max(F_p(q)), 1\leq p\leq P$.
    \item Let $S^K(q, I)$ be the $K^{th}$ highest learned similarity score for the $x\in I$, then the bound of the inaccuracy is $B=F^*(q)-S^K(u, I)$ 
\end{itemize}

\paragraph{Bound of inaccuracy with retrieved items}
Given a set of candidate items retrieved by top $K$ of the dot similarity score:
\begin{itemize}
    \item Input: query $q$, a set of items $G$, $\pi$, $P$ sets of embeddings, $q_p, x_p, 1\leq p\leq P$.
    \item Output: a bound $B_l, B_q$ for the inaccuracy of the retrieved items $x\in G$, such that $B_l\leq s(q, x)\leq B_q$.
    \item For each set of embeddings $p$, let $F_p(q)$ be the $K^{th}$ highest dot similarity score for query $q$.
    \item Thus, the potential highest learned similarity score of the item $x \notin I_p$ will be $F_p(q)$.
    \item Because we do not have the gating weights for $\pi(q, x)$, we can only bound the learned similarity score of $x$ as $S(q, x) \leq \max_p{f'(u_p, x_p)}, 1\leq p\leq P$, where $f'(q_p, x_p)=f(q_p, x_p)$ if $x\in I_p$ otherwise $f'(q_p, x_p)=F_p(u)$.
    \item Thus, for $x\in I$, $B_l = \min_p{f'(q_p, x_p)}, B_q = \max_p{f'(q_p, x_p)}, 1\leq p\leq P$.
\end{itemize}
}

\cut{
\subsection{Heuristics}
\textcolor{blue}{Why is heuristics here? Should we not go to 5.8 directly}

\subsubsection{Per embedding set top K (individual top K)}
Given a value $k_0$, for each embedding set $f_p$, retrieve the top $k_0$ items $I_p$ with the highest dot product for $f_p$. Calculate the learned similarity score for the items in $\mathcal{I}=\cup_{p}I_p$ and return the top $K$ items with the highest learned similiarity score.

\subsubsection{All embedding sets top K (max top K)}
For each item, compute its max dot product over all the embedding set. Given a value $k_0$, return the top $k_0$ items $\mathcal{I}$ with the highest max dot product. Calculate the learned similarity score for the items in $\mathcal{I}$ and return the top $K$ items with the highest learned similarity score.

\subsubsection{Average of dot product top K (average top K)}
For each item, compute the average of its dot product over all the embedding set. Given a value $k_0$, return the top $k_0$ items $\mathcal{I}$ with the highest max dot product. Calculate the learned similarity score for the items in $\mathcal{I}$ and return the top $K$ items with the highest learned similarity score.

\subsubsection{Combined top K}
For each metric $i$ of the dot product, i.e., individual top K, max top K, average top K, retrieve the top $k_i$ items $I_i$. Calculate the learned similarity score for the items in $\mathcal{I}=\cup_{i}I_i$ and return the top $K$ items with the highest learned similiarity score.

Here, instead of calculating the learned similarity score for each metric, we first union the top $k_i$ items from each metric and then calculate the top $K$ w.r.t. the learned similarity.

Examples of the combined top K include individual top K + average top K and max top K + average top K.

Example configurations of $k_i$ include
\begin{itemize}
    \item Given $K$ for top $K$, over retrieve the items with $\lambda\cdot K$.
    \item Retrieve top $\lambda \cdot K$ from average top K.
    \item Retrieve $\lambda\dot K/P$ from each embedding set of individual top K, i.e., P is the number of embeddings. Then we combine the average top K with the individual top K.
    \item Retrieve $\lambda\dot K$ from max top K. Then we combine the average top K with the max top K.
\end{itemize}

We can also have a skewer ratio for $k_i$. For example, we can retrieve $\lambda \cdot K$ from average top K and $\lambda \cdot K/P$ from max top K.
}

\cut{
\subsubsection{Bound analysis}

\todo{Bound of individual top K and max top K based on the max of min top K values}

\todo{Show the bound is tight, i.e., there exists a MoL that can reach the bound}

\todo{Show the bound of max top K is higher than or equal to that of individual top K}

\todo{Bound of average top K combined with individual top K or max top K: Use linear programming}

\todo{Bound of combined top K given the threshold of MoL, i.e., the max value of gating weights is less than $\lambda$}

\subsubsection{Complexity analysis}
\todo{Complexity analysis of heuristics. The dot product is computed on-demand}

\todo{Optimization for cross product embedding sets by preprocessing for average top K and max top K}
}

\cut{
\subsection{Approximate algorithms}

\subsubsection{Naive Top $K$}
The naive top $K$ algorithm is a greedy based algorithm to retrieve the top $K$ items. The key idea is to leverage the dot product scores to scope the set of items for calculating the learned similarity scores. The algorithm works as the following:
\begin{itemize}
\item Retrieve the top $K$ for each set of embeddings by their dot similarity scores
\item Takes the union of the retrieved items as $I$
\item Retrieve the dot similarity scores of all the items in $I$ for each embedding set
\item Calculate the learned similarity score for all the items in $I$.
\item Returns the top $K$ items from $I$ with the highest learned similarity score.
\end{itemize}

\subsubsection{Incremental retrieval with bounded accuracy}
This algorithm leverages the bound analysis in Section~\ref{sec:bound_analysis} to incrementally calculate the learned similarity scores for candidate items until the accuracy of the top $K$ items is acceptable. This algorithm can be used to strike a balance between the quality of the retrieval and its cost.

Given an accuracy threshold $T$ and a cost budget $C$, the algorithm works as the following:
\begin{itemize}
    \item Retrieve the top $K$ for each set of embeddings by their dot similarity scores.
    \item Takes the union of the retrieved items as $I$.
    \item Use the \emph{bound of inaccuracy of retrieved items} to estimate the upper bound of the learned MoL similarity scores.
    \item Take the top $b$ items from with the highest upper bound, u770
\end{itemize}

\subsection{Complexity analysis}
\label{sec:complexity_analysis}

\todo{Try to do some complexity analysis assuming $\pi$ is skewed, i.e., $\pi_{ij}>\lambda$ for some $j$. If this is still not possible, perform the analysis assuming the similarity score is also skewed, i.e., Zipfian distribution}

\subsection{Enforce skewed distribution}
\label{sec:enforce_skew}

\todo{Enforce the skewned distribution of $\pi$ in the training}

\subsection{Physical optimization}
\todo{The key insight is computing gating weights for non-consecutive items is more expensive than consecutive items. So in order to compute items $1, 3, 5$, we can compute the gating weights of items $1, 2, 3, 4, 5$ at similar cost. However, because the gating weights depend on dot products, this also means we need to retrieve the dot product scores of item $2, 4$}
}

%% file: code_algorithm_exact.tex
\begin{algorithm}[h]
\small
    \caption{Exact top $K$ algorithm.}
    \label{code:algorithm:exact}
    \begin{algorithmic}[1]
        \Require query $q$, a set of items $X$, $f_p(\cdot)$, $g_p(\cdot)$ for constructing the component-level embeddings $f_p(q), g_p(x)$
        \Ensure exact top $K$ items
        \State $G \gets \emptyset$
        \For{$p\in P$}
            \State $X_p \gets \{ g_p(x), x \in X \}$ \Comment{Can be preprocessed.} 
            \State $G \gets G \cup TopKDotProduct(f_p(q), X_p)$ \Comment{Retrieve top $K$ items for each pair of embeddings.}
        \EndFor
        \State $S_{min}\gets \infty$
        \For{$x\in G$}
            \State $s \gets MoL(q, x)$
            \If{$s<S_{min}$} $S_{min}\gets s$
            \EndIf
        \EndFor
        \State $G' \gets \emptyset$
        \For{$p\in P$}
            \State $G' \gets G' \cup RangeDotProduct(f_p(q), S_{min}, X_p)$\Comment{Retrieve all items $x \in X_P$ with $\langle f_p(q), x\rangle \geq S_{min}$.}
        \EndFor
        \State \Return $BruteForceTopKMoL(q, G')$ \Comment{Retrieve the top $K$ items from $G'$ with \mol.}
    \end{algorithmic}
\end{algorithm}

%% file: table_example.tex
\begin{table}[t]
    \centering
    \small
    \begin{tabular}{|c|c|c|c|c|c|}
    \hline
        item & $\langle f_1, g_1\rangle$ & $\langle  f_2, g_2 \rangle$ & $\pi_1$ & $\pi_2$ & $\phi$ \\ \hline
        a & 1 & 1 & 0.5 & 0.5 & 1.0 \\ \hline
        b & 0.8 & 0 & 0.5 & 0.5 & 0.4 \\ \hline
        c & 0 & 0.8 & 0.5 & 0.5 & 0.4 \\ \hline
        d & 0.7 & 0 & 1 & 0 & 0.7 \\ \hline
        e & 0.2 & 0.2 & 0.5 & 0.5 & 0.2 \\ \hline
    \end{tabular}
    \vspace{.4em}
    \caption{\changed{Example illustrating how exact- and approximate- top-$k$ algorithms work. We consider a fixed query $q$, and provide inner products $\langle f_p(x), g_p(q)\rangle$s, gating weights $\pi_p$s, and learned similarity scores $\phi$s for that query.}}
    \label{table:example}
    \vspace{-2em}
\end{table}

%% file: code_algorithm_approximate.tex
\begin{algorithm}[t]
\small
    \caption{Approximate top-$k$ algorithms.}
    \label{code:algorithm:approximate}
    \begin{algorithmic}[1]
    \Require a query $q$, a set of items $X$
    \Ensure approximate top $K$ items
    \Function{ApproxTopK}{$q, X, K, K'$}
        \State $G \gets TopKCandidate(q, X, K')$ \Comment{Retrieve the top $K'$ candidates.} 
        \State \Return $BruteForceTopKMoL(q, G, K)$\Comment{Retrieve the top $K$ items \changed{from $G$} with \mol.}
    \EndFunction
    \Require a query $q$, a set of items $X$, $f_p(\cdot)$, $g_p(\cdot)$ for constructing the \changed{$P$} component-level embedding \changed{pairs} $f_p(q), g_p(x)$
    \Ensure union of top $K$ items over $P$ embedding pairs by dot product
    \Function{TopKPerEmbedding}{$q, X, K$}
        \State $G\gets \emptyset$
        \For{$p\in P$}
            \State $X_p \gets \{ g_p(x), x \in X \}$ \Comment{Can be preprocessed.}  
            \State $G\gets G\cup TopKDotProduct(f_p(q), X_p, K)$\Comment{Retrieve the top $K$ items by dot product \changed{for the $p$-th embedding pair.}}
        \EndFor
        \State \Return $G$ \Comment{\changed{Dedup'ed set of the top $K$ item for each of the $P$ queries}}
    \EndFunction
    \Require a query $q$, a set of items $X$, $f_p(\cdot)$, $g_p(\cdot)$ for constructing the component-level embeddings $f_p(q), g_p(x)$
    \Ensure top $K$ items with averaged dot product, $\sum_{p}\langle f_p(q), g_p(x) \rangle/P$ 
    \Function{TopKAvg}{$q, X, K$}
        \State $q'\gets \sum_{p=1}^P f_p(q)$
        \State $X' \gets \{\sum_{p=1}^P g_p(x)/P, x \in X\}$ \Comment{Can be preprocessed.} 
        \State \Return $TopKDotProduct(q', X', K)$
    \EndFunction
    \end{algorithmic}
\end{algorithm}

%% file: sec_evaluation.tex
\section{Evaluation}
\label{sec:evaluation}

In this section, we evaluate the performance of \mol based learned similarity with the proposed load balancing loss, and the efficiency of retrieval algorithms discussed in Section~\ref{sec:algorithm}. \nlprevision{Our code and model checkpoints are available at \githuburl.} 

\cut{
In this section, we evaluate the accuracy and efficiency of our retrieval algorithm. In particular, we want to answer the following questions:

\begin{itemize}
    \item How accurate is our approximate algorithm compared to the exact algorithm and prior work?
    \item How efficient are our exact and approximate algorithms compared to prior work?
    \item What is the tradeoff between the accuracy and efficiency in our approximate algorithm?
    \item How realistic is our complexity analysis (Section~\ref{sec:complexity_analysis})? Given that our complex analysis relies on the distribution of $\pi$ and the learned similarity score, we would like to understand whether our assumptions on the skewness of the distribution hold in practice. \todo{where is complexity analysis?}
    \item What is the impact of enforcing skewness of the distribution of $\pi$ on the performance of the recommendation tasks and the accuracy and efficiency of our algorithms (Section~\ref{sec:enforce_skew})? \todo{do we want to drop this?}
\end{itemize}
}
\vspace{-.2em}
\subsection{Workloads}
\label{sec:evaluation:setup}


\nlprevision{We benchmark MoL with the proposed load balancing loss $\mathcal{L}_{MI}$, on top of state-of-the-art baselines in recommendation systems and question answering. We describe workloads used below.}

\paragraph{Recommendation Systems.} We consider three widely used datasets, the 1M and 20M subsets of MovieLens~\citep{movielens_2015}, and the largest Books subset of Amazon Reviews~\citep{amznreviews_sigir15}. \nlprevision{Sequential retrieval models have been shown to achieve state-of-the-art results on these datasets~\cite{gru4rec_iclr16,sasrec_icdm18,zhai2024actions_icml24}. In these settings, sequential encoders, like RNNs or Transformers, are used to map user representations at time $t$ -- e.g., in a commonly used setting shown in Figure~\ref{fig:mol-adaptation}, the list of items in user history up until time $t$, $\Phi_0, \ldots, \Phi_t$ -- to $\mathbb{R}^d$, and the model is trained to autoregressively predict the next item $x_{t+1}$.
We hence compare MoL with the proposed regularization loss on top of two popular backbones used for sequential retrieval models, SASRec~\citep{sasrec_icdm18} and HSTU~\citep{zhai2024actions_icml24}, against cosine similarity baselines. We utilize user id-based embeddings discussed in Section~\ref{sec:mol-adaptation} and MLPs to parameterize the $P_Q$ query-side and the $P_X$ item-side features.}

\paragraph{Question Answering (QA)} \nlprevision{Natural Questions (NQ)~\cite{naturalquestions_tacl19} is commonly used to evaluate state-of-the-art neural retrieval models, including dense retrieval~\citep{dpr_emnlp20,gtr_emnlp22} and generative retrieval~\cite{genre_iclr21,dsi_neurips22,nci_neurips22,genret_neurips23} approaches in recent years. The most commonly used version~\citep{dsi_neurips22,nci_neurips22,genret_neurips23}, which we reuse in our work, is often referred to as NQ320k. NQ320k consists of 320k
query-items pairs, where the items are from Wikipedia pages and the queries are natural language questions. 
We utilize special aggregation tokens discussed in Section~\ref{sec:mol-adaptation} to parameterize embeddings in MoL, and compare MoL with popular sparse retrieval methods~\cite{bm25_2009,dott5query_2019},
dense retrieval methods~\cite{dpr_emnlp20,sentencet5_acl2022,gtr_emnlp22}, and generative retrieval methods~\cite{genre_iclr21,dsi_neurips22,seal_neurips22,nci_neurips22,genret_neurips23}. Consistent with recent work~\cite{dsiqg_zhuang2023bridginggapindexingretrieval,nci_neurips22,genret_neurips23}, we use the pre-trained query generation model from DocT5Query~\cite{dott5query_2019} to generate synthetic (query, item) pairs for data augmentation.}

Table~\ref{table:exp:workload_stats} summarizes the statistics of these four workloads.
\input{table_workload_stats}

\subsection{Quality of MoL-based Learned Similarity}

\paragraph{Metrics.} \nlprevision{We use Recall (Hit Rate) as the main metric. We report Hit Rate@\{1, 10, 100\} and Mean Reciprocal Rank (MRR) on NQ320K, following~\cite{nci_neurips22,genret_neurips23}, and Hit Rate@\{1, 10, 50, 200\} on \mlonem, \mltwentym, and \amzn, following~\cite{zhai23kdd,zhai2024actions_icml24}.}

\paragraph{Hyperparameter Settings.} We set the weight $\alpha$ for the proposed load balancing loss $\mathcal{L}_{MI}$ to $0.001$ for all experiments. We reuse baseline settings for most other hyperparameters, including learning rate, number of examples used for in-batch negative sampling, etc., with detailed discussions in Appendix~\ref{appendix:experiment-setups}. For the NQ320K dataset, we reuse SEAL~\cite{seal_neurips22} and NCI~\cite{nci_neurips22} results reported by~\cite{nci_neurips22}, and results for other models as reported by~\cite{genret_neurips23}. The Sentence-T5~\cite{sentencet5_acl2022}, GENRE~\cite{genre_iclr21}, DSI~\cite{dsi_neurips22}, SEAL~\cite{seal_neurips22}, DSI+QG~\cite{dsiqg_zhuang2023bridginggapindexingretrieval}, NCI~\cite{nci_neurips22}, and GenRet~\cite{genret_neurips23} rows are all finetuned from T5-base, consistent with MoL, to ensure a fair comparison.
All other results are implemented in PyTorch, and are trained with 1x/2x 48GB GPUs for the recommendation datasets and 4x 80GB GPUs for the QA datasets.

\input{table_ml_model}
\input{table_nlp_retrieval}


\paragraph{Results.} Across the six recommendation scenarios utilizing different sequential encoder backbones, Mixture-of-Logits (MoL rows) consistently outperform \changed{the state-of-the-art dense retrieval baselines (dot products)} by an average of 29.1\% in HR@1, 16.3\% in HR@10, and 18.1\% in MRR (Table~\ref{tbl:model-quality-reco}). On the widely used Natural Questions QA dataset, MoL outperforms all recent generative retrieval approaches as well as strong dense- and sparse- retrieval baselines (Table~\ref{tbl:model-quality-nq320k}). These results validate that learned similarities, in particular MoL, are not only theoretically expressive but also \emph{practically learnable}, improving retrieval quality across heterogeneous scenarios, including sequential retrieval models for Recommendations and finetuning LMs for Question Answering.

\paragraph{Ablation Studies.} We conduct ablation studies for the proposed mutual information-based load balancing loss relative to the best performing method for each dataset (``abl. $\mathcal{L}_{MI}$'' rows). Results show that our proposed $\mathcal{L}_{MI}$ loss improves HR@1 by 2.4\%, HR@10 by 0.8\% and MRR by 1.4\% across the four datasets. In particular, our proposed $\mathcal{L}_{MI}$ loss enables MoL to  outperform the best generative retrieval approach on NQ320K, GenRet~\cite{genret_neurips23}, across all metrics.





\subsection{Top $K$ retrieval performance}
\label{sec:evaluation:perf}

\input{table_topk_v5}

We evaluate the following methods for top $K$ retrieval performance:
\begin{itemize}[leftmargin=*]
    \item Brute-force top $K$ (\topkbf): Evaluate the \mol scores for all items and return the top $K$ items. This is the ground truth in our top $K$ evaluation~\footnote{We omit the baseline with the two-pass exact algorithm (Section~\ref{sec:algorithm:exact}) because the range-based item retrieval can still be expensive when the range threshold is loose. Empirically, the brute-force top $K$ is more efficient on our datasets. We leave the efficient implementation of the two-pass exact algorithm as future work.}.
    \item Per embedding top $K$ (\topknaive{($N$)}): This algorithm is described in Section~\ref{sec:algorithm:approximate}. $N$ is the number of candidate items retrieved from each embedding set, where $N\times P\geq K$.
    \item Average top $K$ (\topkavg{($N$)}): This algorithm is described in Section~\ref{sec:algorithm:approximate}. $N$ is the number of the candidate items retrieved by average dot products, where $N\geq K$.
    \item Combined top $K$ from per embedding top $K$ and average top $K$ (\topkcomb{$N_1$}{$N_2$}): This is described in Section~\ref{sec:algorithm:approximate}. $N_1$ is the number of candidate items retrieved from per embedding top $K$ and $N_2$ is the number of candidate items retrieved from average top $K$, where $N_1\times P+N_2\geq K$.
\end{itemize}

\nlprevision{For each dataset, we evaluate top $K$ retrieval methods based on the best performing model configurations reported in Table~\ref{tbl:model-quality-reco} and Table~\ref{tbl:model-quality-nq320k}.} Table~\ref{tbl:exp:topk} shows the hit rate (HR) and latency of all the methods. The hit rate is normalized by the ground truth, i.e., the hit rate achieved with brute-force top $K$. We measure latency by evaluating a batch of 32 retrieval queries, in order to achieve high accelerator utilization; this is consistent with prior work on GPU/TPU-based retrieval algorithms~\citep{faiss_tbd21,tpuknn_goog_neurips22,zhai23kdd}.
\nlprevision{We omit \mlonem as its size is small (Table~\ref{table:exp:workload_stats}).} 
We set batch size to $32$ for all datasets.
We perform evaluation on a single RTX 6000 Ada GPU. We report latency averaged over 20 warm runs\cut{ for 10\% of the batches}.

We observe that our approximate heuristics achieve high HR with overfetching. For example, \topkavg{500} achieves $>.99$ in relative HR across the board \nlprevision{for 
\mltwentym}, and \nlprevision{\topkavg{100} achieves $>.99$ in relative HR across the board for \nq}. In addition, the combined top $K$ algorithm can outperform both \topknaive{} and \topkavg{} of the corresponding configurations, sometimes significantly, e.g., \topkcomb{5}{200} vs. \topknaive{5} and \topkavg{200} on \amzn. This indicates that the set of candidate items retrieved by each individual approximate algorithm indeed complements each other when the weight distributions, \nlprevision{$\pi_p(q, x)$s,} vary in \mol.

In terms of efficiency, we observe that our approximate heuristics are significantly lower in latency than the exact baselines, especially as the number of items in the dataset becomes large. For example, compared to \topkbf{}, \topkavg{} achieves $>.99$ relative HR@100 with a speedup of \changed{\nlprevision{\nqmaxspeedup in latency for \nq}}. While the algorithm latency grows with the size of the dataset in the brute-force baseline, it grows much slower with the approximate methods. For example, the algorithm latency increases by \nlprevision{\changed{$47.0\times$} from \mltwentym to \amzn in \topkbf{}, while the growth rate is \changed{$10.1\times$} and $1.0\times$ for \topknaive{100} and \topkavg{500},} respectively. Thus, we expect the speedup of the approximate methods to become even more pronounced with larger datasets. 
\cut{
We also notice that \topkavg{} tends to be more efficient than \topknaive{} with comparable HR, e.g., \changed{\topkavg{4000} vs. \topknaive{50} on \amzn with \nlprevision{$10.5\times$}} speedup in latency. 
This is mainly due to two reasons. First, when the HR is comparable, the maximal number of candidate items from \topknaive{} is larger than that of \topkavg{}. Second, compared to \topknaive{}, the computation of \topkavg{} is agnostic to the number of component-level 
embeddings, $P$, because of the materialization optimization described in Section~\ref{sec:algorithm:approximate}. Interestingly, we also see that the combined top $K$ is more efficient than the summation of the latency of its individual components, e.g., \topkcomb{5}{200} is $1.5\times$ faster than the sum of the latency from \topknaive{5} and \topkavg{200} on \mltwentym. This is because our implementation reduces the overhead of the combined method by \nlprevision{consolidating processing shared by the two components}.
}

We also notice that \topkavg{} tends to be more efficient than \topknaive{} at comparable HRs, e.g., \changed{\topkavg{4000} is \changed{$10.5\times$}} faster than \topknaive{50} on \amzn in terms of latency. 
First, the computation in \topkavg{} is agnostic to the number of component-level 
embeddings, $P$, because of the preprocessing described in Section~\ref{sec:algorithm:approximate}. Second, the branching and deduplication steps in \topknaive{} cannot leverage the parallel processing capabilities of GPUs effectively. Additionally, the latency of the combined top $K$ algorithm is lower than the sum of its individual components' latencies; e.g., on \mltwentym, \topkcomb{5}{200} has a latency $1.5\times$ lower than the sum of the individual latencies of \topknaive{5} and \topkavg{200}. \changed{This is done by 
consolidating processing shared by the two components}.

Overall, empirically \topkavg{} strikes a good balance between high HR and low latency, and the combined top $K$ algorithm can be used if the target HR is extremely stringent.

\cut{
\todo{Add the number of items retrieved and costed by \mol in Appendix}
}
\cut{
\todo{Evaluation TODOs (Jiaqi)}
Three base algorithms
\begin{itemize}
    \item Brute force Top K ($BfTopK(K)$): Find the top $K$ items from each embedding set. If there are $P$ embedding sets, the total number of items retrieved (without deduplication) is $P\times K$
    \item Two pass exact Top K ($2pTopK(K)$: Use the two pass algorithm to find the exact top $K$.
    \item Average top K ($AvgK(K)$): Find the top $K$ items with the largest dot product averaged over all embedding sets.
    \item Max top K ($MaxK(K)$): Find the top $K$ items with the largest max dot product over all embedding sets.
\end{itemize}

The combined algorithm
\begin{itemize}
    \item Combine $Alg_1(K_1)$ and $Alg_2(K_2)$ ($Comb(Alg_1(K_1), Alg_2(K_2))$): Take the item set $I_1$ from the base algorithm $Alg_1(K_1)$ and the item set $I_2$ from the base algorithm $Alg_2(K_2)$. Then take the union of $I_1$ and $I_2$, i.e., $I=I_1\cup I_2$
    \item Example: $Comb(TopK(K), AvgK(K))$: Take the items from $TopK(K)$, i.e., $P\times K$ items (without deduplication) and the items from $AvgK(K)$, i.e., $K$ items. 
    \item Example: $Comb(AvgK(K), MaxK(K))$: Take the items from $AvgK(K)$, i.e., $K$ items, and the items from $MaxK(K)$, i.e., $K$ items.
\end{itemize}

Algorithm and parameters evaluated for retrieving the top $K$ items of learned similarity
\begin{itemize}
    \item $K$: The final $K$ items returned by MoL, $K\in \{1, 5, 10, 50, 100\}$
    \item $t$: Multiplier of $K$, $t\in \{1, 5, 10, 50, 100, 200, 300, 400, 500\}$
    \item $TopK(\max(1, \frac{t\times K}{P}))$
    \item $AvgK(t\times K)$
    \item $MaxK(t\times K)$
    \item $Comb(AvgK(\max(1, \frac{t\times K}{2})), TopK(\max(1, \frac{t\times K}{2P})))$, i.e., divide the budget evenly between $TopK$ and $AvgK$.
    \item $Comb(AvgK(\max(1, \frac{t\times K}{2})), MaxK(\max(1, \frac{t\times K}{2}))$, i.e., divide the budget evenly between $AvgK$ and $MaxK$
    \item $Comb(AvgK(t\times K), TopK(\max(1, \frac{t\times K}{20P})))$, i.e., add $1/10$ budget for $TopK$.
    \item $Comb(AvgK(t\times K), MaxK(\max(1, \frac{t\times K}{10}))$, i.e., add $1/10$ budget for $MaxK$.    
\end{itemize}
We compare our technique with two baselines:
\begin{itemize}
    \item \textbf{Union of top $K$ retrieval}: The algorithm retrieves the top $K$ items of each of the $p$ embedding sets based on dot product similarity score. Then it invokes the learned similarity function on the retrieved items to find the final top $K$ items based on the learned similarity function. \textcolor{red}{jz: isn't this exactly the same as the following if we sparsify $\pi$?}
    \item \textbf{Hierarchical retrieval}: ~\cite{zhai23kdd} proposes a hierarchical retrieval algorithm for mixture of logits. \todo{Description of the algorithm}. While this algorithm does not provide any accuracy guarantee, ~\cite{zhai23kdd} shows that the mixture of logits model with hierarchical retrieval outperforms state-of-the-art models in retrieval tasks.
\end{itemize}   

\subsection{Workload analysis}
\todo{Show the distribution of gating weights and dot product scores}
\todo{Figure~\ref{fig:exp:mol_weights}: Distribution of MoL gating weights}

\input{figure_exp_mol_weights}

\todo{Figure~\ref{fig:exp:cosine_scores}: Distribution of Cosine scores}

\input{figure_exp_cosine_scores}

\subsection{Accuracy and efficiency}

\todo{Accuracy vs. Efficiency for exact algorithm, approximate algorithm, and the baselines. Note that the accuracy refers to the accuracy in terms of the learned similarity score function. It does not mean the end-to-end performance of the recommendation retrieval tasks.}

\input{table_recall_latency}

\todo{Efficiency analysis}

\input{table_mol_ratio}

\input{figure_time_breakdown}

\subsection{Other observations}
}

\cut{
\subsection{Skewness analysis}

\todo{Distribution of $\pi$ in practice} 

\todo{Complexity analysis vs. the real processing time, i.e., the growth factor}
}

\cut{
\subsection{Impact of enforcing skewness}

\todo{Impact of enforcing skewness of the distribution on the algorithm and the recommendation tasks. Ideally, we want to show skewed distribution has no impact on the performance of the recommendation tasks.}
}

%% file: table_workload_stats.tex
\begin{table}[t]
\begin{centering}
\small
\begin{tabular}{l|l|l|l|l|l}
\toprule
\multirow{1}{*}{Workload}                  & \multirow{1}{*}{$|Q|$} & \multirow{1}{*}{$|X|$} & \multirow{1}{*}{$|P_q|$} & \multirow{1}{*}{$|P_x|$} & $d_P$ \\ 
\midrule
\mlonem    & 6,040     & 3,649           & 8                        & 4                          & 64         \\   
\mltwentym & 138,493   & 24,186           & 8                        & 4                       & 128            \\ 
\amzn      & 694,897     & 674,044      & 8                        & 8                       & 32               \\
\nq     & 307,373     & 109,739      & 4                        & 4                       & 768 \\
\bottomrule
\end{tabular}
\vspace{.4em}
\caption{Workload statistics.}
\label{table:exp:workload_stats}
\vspace{-3em}
\end{centering}
\end{table}

%% file: table_ml_model.tex
\begin{table}
\ifdef\singlecolumnpdf
\else
\small
\vspace{-.5em}
\begin{center}
  \begin{tabular}{llllll}
    \toprule

%
                          \multirow{2}{*}{\bf Method}  & \multicolumn{4}{c}{\multirow{1}{*}{\bf HR@K}}    & \multirow{2}{*}{\bf MRR}    \\
                                                       & K=1  & K=10  & K=50 & K=200 &  \\
    \midrule
    \multicolumn{6}{l}{\mlonem dataset} \\
SASRec~\cite{sasrec_icdm18}       & .0610     & .2818         & .5470        & .7540     & .1352     \\
                          SASRec + MoL & .0697     & .3036         & .5617        & .7667     & .1441     \\
                         HSTU~\cite{zhai2024actions_icml24}          & .0750     & .3332         & .5956        & .7824     & .1579     \\
                         HSTU + MoL   & \bf .0884 & \bf .3465     & \bf .6022    &  .7935 & \bf .1712  \\
                         HSTU + MoL abl. $\mathcal{L}_{MI}$  & .0847 & .3417  & .6011 & \bf .7942  & .1662  \\
    \midrule
%


    \multicolumn{6}{l}{\mltwentym dataset} \\
                            SASRec~\cite{sasrec_icdm18}       & .0653     & .2883     & .5484      & .7658          & .1375     \\
                            SASRec + MoL &  .0778 & .3102 & .5682 & .7779 & .1535 \\
                            HSTU~\cite{zhai2024actions_icml24}         & .0962     & .3557     & .6146      & .8080          & .1800      \\
                            HSTU + MoL   & \bf .1010 & \bf .3698 & \bf .6260  & \bf .8132      & \bf .1881   \\
                            HSTU + MoL abl. $\mathcal{L}_{MI}$  & .0994  & .3670 & .6241 & .8128 & .1866 \\
    \midrule
    \multicolumn{6}{l}{\amzn dataset} \\
                       SASRec~\cite{sasrec_icdm18}            & .0058      &  .0306     & .0754      & .1431     & .0153     \\
                       SASRec + MoL                           & .0095      & .0429      & .0915      & .1635     & .0212     \\
                       HSTU~\cite{zhai2024actions_icml24}     & .0101      & .0469      & .1066      & .1876     & .0233     \\
                       HSTU + MoL                             & \bf .0156  & \bf .0631  & \bf .1308  & \bf .2173 & \bf .0324 \\
                       HSTU + MoL abl. $\mathcal{L}_{MI}$     & .0153      & .0625      & .1286      & .2172     & .0321     \\
  \bottomrule
\end{tabular}
\vspace{0.4em}
\caption{\nlprevision{Evaluation of performance for sequential retrieval models on MovieLens and Amazon Reviews.}} 
\label{tbl:model-quality-reco}
\vspace{-2em}
\end{center}
\end{table}

%% file: table_nlp_retrieval.tex
\begin{table}
\small
\vspace{-.8em}
\begin{center}
  \begin{tabular}{llllll}
    \toprule
    \multirow{2}{*}{\bf Method}  & \multicolumn{3}{c}{\multirow{1}{*}{\bf HR@K}}    & \multirow{2}{*}{\bf MRR}    \\
      & K=1 &  K=10         & K=100      &     \\
    \midrule
    \multicolumn{5}{l}{\textit{Sparse retrieval}} \\
    BM25~\cite{bm25_2009}              & .297  & .603 & .821 & .402 \\
    DocT5Query~\cite{dott5query_2019}  & .380  & .693 & .861 & .489 \\
    \midrule
    \multicolumn{5}{l}{\textit{Dense retrieval}} \\
    DPR~\citep{dpr_emnlp20}               & .502 & .777 & .909 & .599 \\
    Sentence-T5~\cite{sentencet5_acl2022} & .536 & .830 & .938 & .641 \\
    GTR-Base~\cite{gtr_emnlp22}           & .560 & .844 & .937 & .662 \\
    \midrule
    \multicolumn{5}{l}{\textit{Generative retrieval}} \\
    GENRE~\cite{genre_iclr21} & .552 & .673 & .754 & .599 \\
    DSI~\cite{dsi_neurips22} & .552 & .674 & .780 & .596 \\
    SEAL~\cite{seal_neurips22} & .570 & .800 & .914 & .655 \\ 
                                      DSI+QG~\cite{dsiqg_zhuang2023bridginggapindexingretrieval} & .631 & .807 & .880 & .695 \\
                                      NCI~\cite{nci_neurips22}   & .659 & .852 & .924 & .731 \\ 
                                      GenRet~\cite{genret_neurips23}   & .681 & .888 & .952 & .759  \\ 
    \midrule
    \multicolumn{5}{l}{\textit{Learned similarities}} \\
    MoL       & \bf .685      &  \bf .919     & \bf .970  & \bf .773 \\
    MoL abl. $\mathcal{L}_{MI}$ & .673  & \bf .919 & .968  & .767 \\
  \bottomrule
\end{tabular}
\vspace{0.4em}
\caption{\nlprevision{Evaluation of performance for QA retrieval models finetuned from language models on Natural Questions.}} 
\label{tbl:model-quality-nq320k}
\vspace{-3em}
\end{center}
\end{table}

%% file: table_topk_v5.tex
\begin{table*}[ht]
\small 
\vspace{-1em}
\begin{center}
  \begin{tabular}{clllllll}
    \toprule

                        & \bf Method  & \bf HR@1 & \bf HR@5 & \bf HR@10 & \bf HR@50 & \bf HR@100 &  \bf Latency / ms \\                        
    \midrule
    \cut{
\multirow{10}{*}{\mlonem} & \topkbf & 1.00  &	1.00  &	1.00  &	1.00  &	1.00  &	.67$\pm$.03 \\
                       & \topknaive{5}  & .838	& .771	& .683	& .499	& .437	& 1.02$\pm$.05 \\
                       & \topknaive{10}   & .932	& .908	& .844	& .678	& .618	& 1.03$\pm$.05\\
                       & \topknaive{50}   & \bf 1.00	& \bf 1.00	& \bf 1.00	& .978	& .954	& 1.03$\pm$.04\\
                       & \topknaive{100}   & \bf 1.00	& \bf 1.00	& \bf 1.00	& \bf .995	& \bf .990	& 1.30$\pm$.03\\
                       & \topkavg{200}   &  .980	& .974	& .970	& .958	& .948	& .82$\pm$.04\\
                       & \topkavg{500}   & \bf .992	& \bf .990	& \bf .991	& \bf .990	& \bf .990	& .80$\pm$.03\\
                       & \topkavg{1000}   &  \bf .996	& \bf .998	& \bf 1.00	& \bf .997	& \bf .999	& .79$\pm$.04\\
                       & \topkcomb{5}{200}  & \bf 1.00	& \bf 1.00	& \bf 1.00	& .997	& .999	& 1.18$\pm$.04 \\
    \midrule

}
\multirow{8}{*}{\mltwentym} & \topkbf & 1.00 & 1.00 & 1.00 & 1.00 & 1.00 & 2.73$\pm$0.01 \\

                       & \topknaive{5}  & 0.69          & 0.67          & 0.63          & 0.46          & 0.40          & 1.61$\pm$0.05 \\

                       & \topknaive{10}   & 0.95          & 0.88          & 0.82          & 0.65          & 0.57          & 1.65$\pm$0.06 \\

                       & \topknaive{50}   & \textbf{1.00} & \textbf{1.00} & \textbf{0.99} & 0.95          & 0.92          & 1.64$\pm$0.05 \\

                       & \topknaive{100}   &	\textbf{1.00} & \textbf{1.00} & \textbf{1.00} & \textbf{0.99} & 0.98          & 2.31$\pm$0.02 \\

                       & \topkavg{200}   &	\textbf{1.00} & \textbf{1.00} & \textbf{1.00} & \textbf{0.99} & 0.97          & 1.19$\pm$0.04 \\

                       & \topkavg{500}   &  \textbf{1.00} & \textbf{1.00} & \textbf{1.00} & \textbf{1.00} & \textbf{1.00} & 1.22$\pm$0.05 \\
                       & \topkcomb{5}{200}    & \textbf{1.00} & \textbf{1.00} & \textbf{1.00} & \textbf{1.00} & \textbf{1.00} & 1.86$\pm$0.06 \\
    \midrule

\multirow{16}{*}{\amzn} & \topkbf & 1.00 & 1.00 & 1.00 & 1.00 & 1.00 & 128.36$\pm$0.30 \\
                       & \topknaive{5}  & \textbf{1.00} & 0.92          & 0.90          & 0.61          & 0.48          & 20.47$\pm$0.07 \\
                       & \topknaive{50}    & \textbf{1.00} & \textbf{1.00} & \textbf{1.00} & 0.98          & 0.93          & 21.60$\pm$0.08 \\
                       & \topknaive{100}  & \textbf{1.00} & \textbf{1.00} & \textbf{1.00} & \textbf{0.99} & 0.97          & 23.32$\pm$0.08 \\
                       & \topknaive{200}  & \textbf{1.00} & \textbf{1.00} & \textbf{1.00} & \textbf{1.00} & \textbf{1.00} & 26.55$\pm$0.12 \\
                       & \topkavg{200}   	& 0.97 & 0.92 & 0.91 & 0.76 & 0.67          & 1.13$\pm$0.04 \\
                       & \topkavg{500}   	& 0.97 & 0.98 & 0.97 & 0.86 & 0.81 & 1.17$\pm$0.04 \\
                       & \topkavg{1000}   	& \textbf{0.99} & 0.98 & \textbf{1.00} & 0.92 & 0.88 & 1.12$\pm$0.05 \\
                       & \topkavg{2000}  	& \textbf{1.00} & \textbf{0.99} & \textbf{1.00} & 0.95          & 0.92          & 1.20$\pm$0.02 \\
                       & \topkavg{4000}  	& \textbf{1.00} & \textbf{1.00} & \textbf{1.00} & 0.96          & 0.95          & 2.05$\pm$0.01 \\
                       & \topkavg{8000}       & \textbf{1.00} & \textbf{1.00} & \textbf{1.00} & 0.97          & 0.97          & 3.79$\pm$0.01 \\
                       & \topkcomb{5}{200}  	& \textbf{1.00} & \textbf{1.00} & \textbf{1.00} & 0.96          & 0.95          & 20.75$\pm$0.07 \\
                       & \topkcomb{50}{500}   	& \textbf{1.00} & \textbf{1.00} & \textbf{1.00} & \textbf{0.99} & 0.96          & 22.12$\pm$0.07 \\
                       & \topkcomb{100}{1000}   & \textbf{1.00} & \textbf{1.00} & \textbf{1.00} & \textbf{0.99} & 0.98          & 24.02$\pm$0.13 \\
                       & \topkcomb{200}{2000} & \textbf{1.00} & \textbf{1.00} & \textbf{1.00} & \textbf{1.00} & \textbf{1.00} & 28.01$\pm$0.11 \\
    \midrule

\multirow{6}{*}{\nq} & \topkbf & 1.00	& 1.00	& 1.00	& 1.00	& 1.00	& 37.74$\pm$.47\\
                       & \topknaive{5}  & \bf 1.00	& \bf 1.00	& \bf 1.00	& 0.96	& \bf 1.00	& 4.71$\pm$0.08\\
                       & \topknaive{10}   & \bf 1.00	& \bf 1.00	& \bf 1.00	& 0.98	& \bf 1.00	& 4.83$\pm$0.08\\
                       & \topknaive{50}   & \bf 1.00	& \bf 1.00	& \bf 1.00	& \bf 1.00	& \bf 1.00	& 6.31$\pm$0.09\\
                       & \topkavg{100}   &	\bf 1.00	& \bf 1.00	& \bf 1.00	& \bf 1.00	& \bf 1.00	& 0.57$\pm$0.05\\
                       & \topkcomb{5}{100}    & \bf 1.00	& \bf 1.00	& \bf 1.00	& \bf 1.00	& \bf 1.00	& 5.28$\pm$0.08 \\


  \bottomrule
\end{tabular}
\vspace{0.3em}
\caption{\nlprevision{Evaluation of top $K$ retrieval performance, with hit rate (HR) normalized by the brute-force top $K$ method and latency 
measured over 
a batch of queries (where the batch size is 32). (Relative) hit rate higher than .99 is marked in \textbf{bold}.}}
\label{tbl:exp:topk}
\vspace{-2em}
\end{center}
\end{table*}

%% file: figure_exp_mol_weights.tex
\begin{figure*}[t]
    \caption{\todo{Distribution of \mol gating weights}}
    \label{fig:exp:mol_weights}
    \centering
    \begin{subfigure}[t]{0.32\textwidth}
        \centering
        \includegraphics[width=.95\linewidth]{figure/placeholder_5x8.png}
        \caption{Workload \mlonem}
        \label{fig:exp:mol_weights:ml1m}
    \end{subfigure}%
    \begin{subfigure}[t]{0.32\textwidth}
        \centering
        \includegraphics[width=.95\linewidth]{figure/placeholder_5x8.png}
        \caption{Workload \mltwentym}
        \label{fig:exp:mol_weights:ml20m}
    \end{subfigure}%
    \begin{subfigure}[t]{0.32\textwidth}
        \centering
        \includegraphics[width=.95\linewidth]{figure/placeholder_5x8.png}
        \caption{Workload \amzn}
        \label{fig:exp:mol_weights:amzn}
    \end{subfigure}%
\end{figure*}

%% file: figure_exp_cosine_scores.tex
\begin{figure*}[t]
    \caption{\todo{Distribution of Cosine similarity scores}}
    \label{fig:exp:cosine_scores}
    \centering
    \begin{subfigure}[t]{0.32\textwidth}
        \centering
        \includegraphics[width=.95\linewidth]{figure/placeholder_5x8.png}
        \caption{Workload \mlonem}
        \label{fig:exp:cosine_scores:ml1m}
    \end{subfigure}%
    \begin{subfigure}[t]{0.32\textwidth}
        \centering
        \includegraphics[width=.95\linewidth]{figure/placeholder_5x8.png}
        \caption{Workload \mltwentym}
        \label{fig:exp:cosine_scores:ml20m}
    \end{subfigure}%
    \begin{subfigure}[t]{0.32\textwidth}
        \centering
        \includegraphics[width=.95\linewidth]{figure/placeholder_5x8.png}
        \caption{Workload \amzn}
        \label{fig:exp:cosine_scores:amzn}
    \end{subfigure}%
\end{figure*}

%% file: table_mol_ratio.tex
\begin{table*}[t]
\centering
\caption{\todo{Percentage of items scored by MoL for \amzn}}
\label{table:exp:mol_ratio}
\centering
\begin{tabular}{|l|lll|lll|lll|lll|lll|}
\hline
k   & \multicolumn{3}{c|}{1} & \multicolumn{3}{c|}{5} & \multicolumn{3}{c|}{10} & \multicolumn{3}{c|}{50} & \multicolumn{3}{c|}{100} \\
\hline
multiplier & 1 & 5 & 10 & 1 & 5 & 10 & 1 & 5 & 10 & 1 & 5 & 10 & 1 & 5 & 10    \\
\hline
\topkunion          &  0.00 & 0.00  & 0.00 &  0.00 & 0.00   &  0.00 &  0.00 & 0.00  &  0.00 &  0.00 & 0.00  &  0.00 &  0.00 & 0.00  &  0.00\\
\hline
\topkmax          &  0.00 & 0.00  & 0.00 &  0.00 & 0.00   &  0.00 &  0.00 & 0.00  &  0.00 &  0.00 & 0.00  &  0.00 &  0.00 & 0.00  &  0.00\\
\hline
\topkavg          &  0.00 & 0.00  & 0.00 &  0.00 & 0.00   &  0.00 &  0.00 & 0.00  &  0.00 &  0.00 & 0.00  &  0.00 &  0.00 & 0.00  &  0.00\\
\hline
\topkavgunion          &  0.00 & 0.00  & 0.00 &  0.00 & 0.00   &  0.00 &  0.00 & 0.00  &  0.00 &  0.00 & 0.00  &  0.00 &  0.00 & 0.00  &  0.00\\
\hline
\topkavgmax          &  0.00 & 0.00  & 0.00 &  0.00 & 0.00   &  0.00 &  0.00 & 0.00  &  0.00 &  0.00 & 0.00  &  0.00 &  0.00 & 0.00  &  0.00\\
\hline
\end{tabular}
\end{table*}

%% file: figure_time_breakdown.tex
\begin{figure*}[t]
    \caption{\todo{Time breakdown of top $K$ algorithms on \amzn}}
    \label{fig:exp:timebreakdown}
    \centering
    \begin{subfigure}[t]{0.32\textwidth}
        \centering
        \includegraphics[width=.95\linewidth]{figure/placeholder_5x8.png}
        \caption{$multiplier = 1$}
        \label{fig:exp:timebreakdown:m1}
    \end{subfigure}%
    \begin{subfigure}[t]{0.32\textwidth}
        \centering
        \includegraphics[width=.95\linewidth]{figure/placeholder_5x8.png}
        \caption{$multiplier = 5$}
        \label{fig:exp:timebreakdown:m5}
    \end{subfigure}%
    \begin{subfigure}[t]{0.32\textwidth}
        \centering
        \includegraphics[width=.95\linewidth]{figure/placeholder_5x8.png}
        \caption{$multiplier = 10$}
        \label{fig:exp:timebreakdown:m10}
    \end{subfigure}%
\end{figure*}

%% file: sec_related_work.tex
\vspace{-.2em}
\section{Related work}
\label{sec:related_work}
\vspace{-.1em}

\paragraph{Similarity Functions in Retrieval.} Most information retrieval models in \nlprevision{recommendation systems and natural language processing (e.g., question answering)} follow a classical two-stage paradigm~\citep{ytdnn_goog_recsys16,dpr_emnlp20}, where up to billions of items~\citep{pixie_pins_www18,mind_baba_cikm19,zhai23kdd,borisyuk2024linr_cikm24} are first filtered down to hundreds in the \textit{retrieval} stage, followed by \nlprevision{another stage (e.g., ranking in recommendation systems or generation in RAG~\cite{rag_neurips20})} that produces the final results. Earlier work on large-scale neural retrieval models primarily utilize dual-encoder (dense retrieval, etc.) setups, with dot products as the similarity function~\citep{ytdnn_goog_recsys16,dpr_emnlp20,sentencet5_acl2022,gtr_emnlp22}. Researchers quickly realized that dot products limited retrieval stage's performance, and explored various learned similarity-based approaches. Prominent variants include maximum similarity based on multiple embeddings~\citep{mind_baba_cikm19,colbert_sigir20,colbertv2_naacl22}, specialized neural networks, often leveraging Hadamard products~\citep{ncf_www17,nnm_wsdm20,flashlight_log22,borisyuk2024linr_cikm24}, and representing item ids as token sequences (``learned index structures''), either implicitly defined during tree traversal~\citep{plt_icml16,tdm_kdd18,otm_icml20} or explicitly in the ``generative retrieval'' setups~\citep{genre_iclr21,dsi_neurips22,nci_neurips22,seal_neurips22,dsiqg_zhuang2023bridginggapindexingretrieval,genret_neurips23}. It has been shown, however, that learned neural distances often fail to outperform dot products, e.g., Hadamard MLPs in recommendation systems~\citep{ncf_mf_goog_recsys20} and DSI for QA scenarios in NLP~\citep{genret_neurips23}. Learned index structures further introduce stability and latency challenges as both NLP and recommendation systems need to support billion-scale realtime updated set of items~\citep{pixie_pins_www18,zhai23kdd,borisyuk2024linr_cikm24}.  Despite these challenges, significant gains (17\% gains at Hit Rate@100~\citep{zhai23kdd} to 24\% gains at Hit Rate@400~\citep{borisyuk2024linr_cikm24}) with learned similarities have been reported in recent years; these can be attributed to careful construction of learned similarity functions~\citep{colbertv2_naacl22,zhai23kdd}, implicit diversification done as part of beam search~\citep{dr_cikm21}, explicit incorporation of side-information using special neural architectures~\citep{zhai23kdd, borisyuk2024linr_cikm24}, and hardware-aware similarity function and inference algorithm design on GPUs~\citep{tpuknn_goog_neurips22,zhai23kdd,cagra_icde24, borisyuk2024linr_cikm24}.


\paragraph{Load Balancing for Conditional Computations in Neural Networks} \nlprevision{Conditional computations have been widely utilized in deep learning models~\cite{bengio2016conditionalcomputationneuralnetworks,smoe_iclr17,modsquad_cvpr23}. Regularization losses have been proposed based on the observation that an ideal policy should evenly utilize all compute units in aggregate while being sparse at an individual example level~\cite{bengio2016conditionalcomputationneuralnetworks}. Mixture-of-experts, a common way to implement conditional computations, has been widely used in language and vision domains~\cite{modsquad_cvpr23,smoe_iclr17} where mutual information-based regularization losses between experts and tasks~\cite{modsquad_cvpr23} and experts and tokens~\cite{moduleformermoe_2023} have been shown to help with various architectures.}

\paragraph{Efficient Nearest Neighbor Search (NNS)} NNS has been a popular research topic due to their critical role in large-scale retrieval and vector databases. Most studies focus on the dot product case, also known as Maximum Inner Product Search (MIPS). Various techniques were proposed and analyzed, including tree structures~\citep{kdtree_1975,conetree_mips_kdd12}, locality sensitive hashing~\citep{lsh_vldb99,alsh_ping_neurips2014}, production quantization~\citep{pq_nns_pami11,quant_mips_kumar_aistats16}, data partitioning~\citep{clusteringss_tkde02,hdss_sigmod11},
graph-based methods~\citep{hnsw_tpami18,diskann_neurips19}, and so on. The general case for NNS utilizing learned similarities remains less studied; for learned index structures, techniques to construct trees 
have been proposed to ensure beam search result in globally optimal top-$K$ results~\citep{otm_icml20}. Algorithms based on implicit~\citep{hnsw_tpami18,diskann_neurips19,flashlight_log22,cagra_icde24} or explicit graphs~\citep{flashlight_log22} have been proposed to obtain a tractable candidate set in multi-stage retrieval setups; however, such approaches' performance can degrade when the similarity function is not a metric, and constructing appropriate graph indices for non-metric similarity functions can remain challenging even for the inner product case~\citep{yandex2018neurips}. Due to GPUs and other accelerators having orders of magnitude higher arithmetic intensity vs CPUs, traditional quantization techniques~\citep{alsh_ping_neurips2014,quant_mips_kumar_aistats16} no longer fully utilize the compute; accelerator-specific nearest neighbor
algorithms that benefit from increased compute have been proposed recently~\citep{tpuknn_goog_neurips22,zhai23kdd,cagra_icde24,borisyuk2024linr_cikm24}.

\cut{
\todo{Bailu: add some vectordb references?}
}

\cut{
Multi-vector vector search
\begin{itemize}
    \item For multi-modal vector search.
    \item They assume the weights are fixed and given.
    \item The query can request the search based on a subset of the vectors.
    \item The number of embedding sets is small, i.e., 2~\cite{vbase23osdi, milvus21sigmod} or up to 4~\cite{wang2023must}.
\end{itemize}
}

%% file: sec_conclusion.tex
\vspace{-.2em}
\section{Conclusion}
We have analyzed techniques for efficient retrieval with expressive learned similarities in this work. We begin by showing Mixture-of-Logits (\mol) is a universal approximator of \changed{all} similarity functions, and further empirically learnable -- \mol  with our proposed load balancing loss consistently outperforms dot products (dense retrieval),
sparse retrieval, and generative retrieval approaches across
Recommendation Systems and Question Answering scenarios, setting new state-of-the-art across common, heterogeneous benchmark datasets. We next propose exact and approximate algorithms to enable efficient retrieval using learned similarity functions, and show their correctness and \changed{error} bounds. Our approximate top $K$ algorithms can reach \changed{$>.99$} of Hit Rate relative to exact algorithms, while achieving up to \maxspeedup reduction in end-to-end latency and with minimal indexing overheads. We expect the speedups to be further amplified with larger-scale datasets and GPU kernel optimizations.  
Given MoL's \changed{impressive empirical} performance gains of 20\%-30\% across Hit Rate@50-400 over hundreds of millions to billions of items~\citep{zhai23kdd,borisyuk2024linr_cikm24} and broad applicability across heterogeneous scenarios, our work provides strong theoretical and practical justifications for migrating web-scale vector databases away from dense retrieval and MIPS to Retriev\underline{a}l w\underline{i}th Learned Similarities (RAILS) on GPUs.



%% file: sec_appendix.tex
\section{Experiment Setups}
\label{appendix:experiment-setups}

\subsection{Reproducibility}

The implementations and hyperparameter settings for reproducing our experiment results can be found at \githuburl. We discuss specific details below.

\subsection{Parameterization of low-rank (``component-level'') embeddings} 

\label{sec:app-exp-emb-parameterization-qa}

In this section, we elaborate on the embedding parameterization methods for MoL that we discussed in Section~\ref{sec:mol-adaptation}.

\subsubsection{Recommendation Systems} \nlprevision{Prior work have shown that careful parameterization of low-rank (``component-level'') embeddings, or $f_p(q)$ and $g_p(x)$s for $1 \leq p \leq P$, can significantly improve MoL's performance~\cite{borisyuk2024linr_cikm24}. In the context of large-scale recommendation systems, cluster information based on interests of cohorts of members and topics of posts by themselves can lead to 10\% recall gain at $K=400$~\cite{borisyuk2024linr_cikm24}. However, we cannot easily access similar information in the publicly available MovieLens~\cite{movielens_2015} and Amazon Reviews~\cite{amznreviews_sigir15} datasets. We therefore follow implementation provided by~\citep{zhai23kdd} and additionally optionally utilizes a User ID keyed one-hot embedding as one query-side low-rank (``component-level'') embeddings $f_p(q)$, which is a widely used technique in recommendation systems~\citep{matrix-factorization-netflix09} that we discussed in Section~\ref{sec:mol-adaptation}. All other component-level embeddings, $f_p(q)$s and $g_p(x)$s, are obtained by applying a multi-layer
perceptron (MLP) on top of query-side/item-side representations in standard sequential recommendation setups~\citep{gru4rec_iclr16,sasrec_icdm18}. The overall setup is illustrated on the right hand side of Figure~\ref{fig:mol-embedding-construction}.}

\input{figure_mol_embeddings}

\subsubsection{Question Answering (QA)} 

\nlprevision{Unlike Recommendation Systems, retrieval models used in question answering generally take the full semantic representation(s) of the query and/or the document as input, and are finetuned on top of pre-trained language models with homogeneous inputs, or wordpiece / sentencepiece tokens. Our MoL embedding construction consists of two components, special aggregation tokens and parameterized pooling. We present embedding construction on the query side first.}

\paragraph{Special Aggregation Tokens} \nlprevision{Given both queries and documents are represented as token sequences (e.g., SentencePieces~\cite{sentencepiece_emnlp18} in T5~\cite{t5_raffel2023exploringlimitstransferlearning}), we propose to add special tokens that can be used to aggregate different aspects of information as part of the overall self-attention based language model. Specifically, on the query side, let the tokenized sequence be $SP_1, SP_2, \ldots, SP_N$. During finetuning of the pretrained language model, we create $P_Q$ special tokens, $Q_1, \ldots, Q_{P_Q}$, and add them to the vocabulary of the query tokenizer. We also append those exact same $P_Q$ tokens before $SP_1, SP_2, \ldots, SP_N$, so that the $P_Q$ special tokens can be used to aggregate information across the query input using early-fusion mechanisms. Note that many question answering scenarios~\cite{dpr_emnlp20,genre_iclr21,nci_neurips22,gtr_emnlp22,genret_neurips23} utilize bidirectional language models for retrieval, like BERT~\cite{bert_naacl19} or T5~\cite{t5_raffel2023exploringlimitstransferlearning}; for recent unidirectional language models, we can add the special aggregation tokens $X_1, \ldots, X_{P_X}$ and $Q_1, \ldots, Q_{P_Q}$ to the end of the input sequence instead. Our construction can also be viewed as a way to extend the CLS token in BERT~\cite{bert_naacl19,chang-etal-2023-multi_acl23} to cover multiple aspects of information, in a way that encourages diversity via the $\mathcal{L}_{MI}$ load balancing loss discussed in Section~\ref{sec:mol}.}

\paragraph{Parameterized Pooling} \nlprevision{We next add a pooling layer after the language model to encourage learning of aggregation mechanisms separate from language semantics. For each position $1 \leq p \leq P_Q$, this pooling layer defines a probability distribution over different positions in language model's outputs, or $(0, \ldots, max\_seq\_len-1)$. We further \emph{parameterize} the pooling layer, using the $D$-dimensional embedding at the first position after encoders. This enables us to define a pooling policy, at an example-level, how to weight each of the $max\_seq\_len$ LM encoder outputs to arrive at the $P_Q$ MoL embeddings.}

\vspace{.3em}
\nlprevision{The embedding construction on the item-side is identical. We illustrate the overall finetuning setup we use for question answering on the left hand side of Figure~\ref{fig:mol-embedding-construction}.}

\subsection{Parameterization of $\pi_{p}(q, x)$ matrices}

We follow the implementation provided in the original MoL paper~\citep{zhai23kdd}\cut{~\footnote{\url{https://github.com/facebookresearch/generative-recommenders}}}, which parameterizes $\pi_{p}(q, x)$ as a two-layer multi-layer perceptron (MLP) with SiLU~\citep{elfwing2017silu} non-linearity. For recommendation datasets (\mlonem, \mltwentym, \amzn), the inputs to this MLP consist of user-side features, item-side features, and the $P$ dot products $\langle f_p(q), g_p(x)\rangle$s between the low-rank embeddings. For question answering datasets (NQ320K), we only use the last part -- the $P$ dot products $\langle f_p(q), g_p(x)\rangle$s between the low-rank embeddings -- as inputs to this MLP.

\subsection{Hyperparameter settings} 

\subsubsection{Recommendation Systems}
\nlprevision{We use an identical number of sampled negatives for dot product baselines (cosine similarity, ``SASRec'', ``HSTU'' rows in Table~\ref{tbl:model-quality-reco}) and Mixture-of-Logits (``SASRec + MoL'', ``HSTU + MoL'' rows in Table~\ref{tbl:model-quality-reco}) to ensure a fair comparison, which is $128$ for ML-1M and ML-20M and $512$ for Amazon Books following prior work. For ``+ MoL'' rows, we additionally grid searched $|P_x|$ in $\{2, 4, 8, 16\}$, $d_P$ in $\{32, 64, 128\}$, whether to enable user-id based learned embeddings, and the dropout rate to apply to user-id based embeddings in $\{0.2, 0.5, 0.8\}$ for the smaller MovieLens datasets. We followed initial hyperparameters provided by the authors~\citep{zhai23kdd} for all other parameters. The models are trained using PyTorch over 1 NVIDIA RTX 6000 Ada GPU for the smaller \mlonem and \mltwentym datasets and 2 NVIDIA RTX 6000 Ada GPUs for the larger \amzn datasets.}

\subsubsection{Question Answering (QA)} 

\nlprevision{We train the model with AdamW optimizer~\cite{loshchilov2018decoupled_iclr19}, and grid searched learning rate in \{2e-4, 5e-4, 8e-4\} due to the introduction of the parameterized pooling component (Appendix~\ref{sec:app-exp-emb-parameterization-qa}). We apply linear scheduling with warm-up over a fixed 10\% of the training epochs. We train the model on 4 NVIDIA H100 80GB GPUs with a local batch size of 512. Note that due to the computational requirements of this dataset, prior work are frequently trained on 8 GPUs~\cite{dpr_emnlp20,nci_neurips22} or more, e.g., 32 GPUs in GENRE~\cite{genre_iclr21} and 256 TPUs in DSI~\cite{dsi_neurips22}. We perform in-batch negative sampling, consistent with baselines~\cite{dpr_emnlp20,sentencet5_acl2022}. For MoL hyperparameters, we grid searched $P_Q$ and $P_X$ in \{(2, 2), (4, 4), (8, 8), (16, 16)\}, kept $d_P$ identical to the embedding dimension of the pretrained language model ($768$), and selected the best hyperparameters utilizing a validation set.}

\changed{
\section{Examples for exact and approximate top $K$ retrieval algorithms}
\label{appendix:example}

We provide examples for our exact- and approximate- retrieval algorithms discussed in Section~\ref{sec:algorithm} to facilitate understanding.

\paragraph{Retrieval with Exact Algorithm (Section~\ref{sec:algorithm:exact}).}
\input{table_example}
Table~\ref{table:example} shows an example with $4$ items and $2$ pairs (groups) of embeddings ($P=2$). Assume the goal is to retrieve the top $K=2$ items. In the first stage, we retrieve the top $2$ items for each embedding set, i.e., item $a, b, c$ are retrieved. For each of the retrieved items, we calculate the MoL scores based on their gating weights, i.e., $1.0, 0.4, 0.4$ for $a, b$, and $c$, respectively. Here, $S_{min}$ is set to $0.4$. In the second stage, we retrieve all items with $\langle f, g\rangle \geq 0.4$ for each embedding set, i.e., item $d$, and then calculate their corresponding MoL score, i.e., $0.7$. The algorithm returns $a$ and $d$ as the top $2$ items.

\paragraph{Retrieval with Approximate Algorithm (Theorem~\ref{theorem:topkgap}).}
    Consider again the example shown in Table~\ref{table:example}. Assume we want to retrieve the top $2$ items. With top $K$ over component-level embeddings, item $a, b, c$ are retrieved, and $S_k=0.4$. Since $S$ is calculated as $\max\{0.7, 0.2\}$, the upper bound of the gap is $0.3$. Here, the gap is exact, i.e., the actual second largest MoL score with item $d$ is $0.3$ higher than the second largest MoL score from the set of retrieved items $\{a, b, c\}$. If we bound the gap with the retrieved items only, i.e., $a, b, c$, then we will get a looser bound of $0.8-0.4=0.4$.
}

\cut{
\section{Ablation study}

\subsection{Impact of multipliers}
\todo{Show the change of recall with additional multipliers on \amzn}

\todo{Table~\ref{fig:exp:multiplier}}

\input{table_multiplier}

\subsection{Impact of combination ratio}

\todo{Show the result with different ratio of combined top k for one selected $k$}

\todo{Table~\ref{fig:exp:comb_ratio}}

\input{table_comb_ratio}
}

\cut{
\section{Intro Outline}

\point{
Introduction of vector databases
\begin{itemize}
    \item Introduction of vector database
    \item APIs of vector databases
\end{itemize}
}

\point{
Introduction of learned similarity
\begin{itemize}
    \item State-of-the-art recommendation systems now move to learned similarity function
    \item Examples of learned similarity: Tree work, Path work, beam search, Mixture of logits
    \item Figure of the learned similarity: Figure~\ref{fig:learned_similarity}
\end{itemize}
}

\point{
How learned similarities are computed
\begin{itemize}
    \item Add a very brief introduction of the related work for state-of-the-art recommendation models.
    \item Add a description of the related work for learned similarity, i.e., tree, path, MoL.
    \item Three figures to show the high-level learned similarity for the related work, i.e., tree, path, MoL. The three figures will be put in one row.
    \item The math formulation of the four learned simliarities.
    \item How the four learned similarities are computed in practice?
\end{itemize}
}

\point{
Challenge to support learned similiarity in vector databases
\begin{itemize}
    \item Property of learned similarity function: Expensive to compute, not a distance measure (to verify), cannot be precomputed offline.
    \item Learned similarity varies in their forms / formulation.
    \item How learned similarity function is used: Compute on-demand
    \item Distances supported by vector databases: dot product based
\end{itemize}
}

\point{
How we tackle this problem?
\begin{itemize}
    \item Describe the rationale of using rank as the measurement of expressiveness
    \item We step back to the motivation of having the various learned similarity methods: increase the expressiveness of the similarity
    \item Explain the metric for expressiveness, i.e., rank
    \item Explain why dot product similiarty is not expressive?
    \item Our key insight: Among learned similarties, we show that MoL is a universal approximator, i.e., it can approximate full rank. Thus, it is the most general form of learned similarity.
\end{itemize}
}

\point{
This work
\begin{itemize}
    \item Goal: Support efficient retrieval of items with learned similarity functions
    \item Key insight: MoL is expressive enough to represent any learned similarity function
    \item Key insight: For MoL, we can leverage dot product similarity function to retrieve items with learned similarity function \textcolor{blue}{which allows us to partially reuse existing vector db apis?}
\end{itemize}
}

\section{Hardware Optimization Considerations}

Specific problem setup for the efficient implementation
\begin{itemize}
    \item The item embeddings are cached in memory (HBMs on GPUs or DRAMs on CPUs). The cost of the top $k$ algorithm is generally memory-bandwidth bound.
    \item Access the embeddings of non-consecutive items is very expensive unless the number of items accessed is orders of magnitude less than the total number of items. 
    \item The computation of dot product and top $k$ for dot product is much faster than scanning the embeddings in the memory. \textcolor{blue}{for the time needed to read one byte, we can perform 300-2000x multiplications}
    \item The computation of the learned similarity is relatively more expensive. 
    \item The computation of the learned similarity needs to be batched, e.g., the number of (user, item) pairs per batch can be \textcolor{blue}{$8\cdot 10^4 \sim \cdot 10^5$}.
\end{itemize}

}

%% file: figure_mol_embeddings.tex
\begin{figure*}[t]
    \vspace{-1em}
    \centering
    \includegraphics[width=0.9\linewidth]{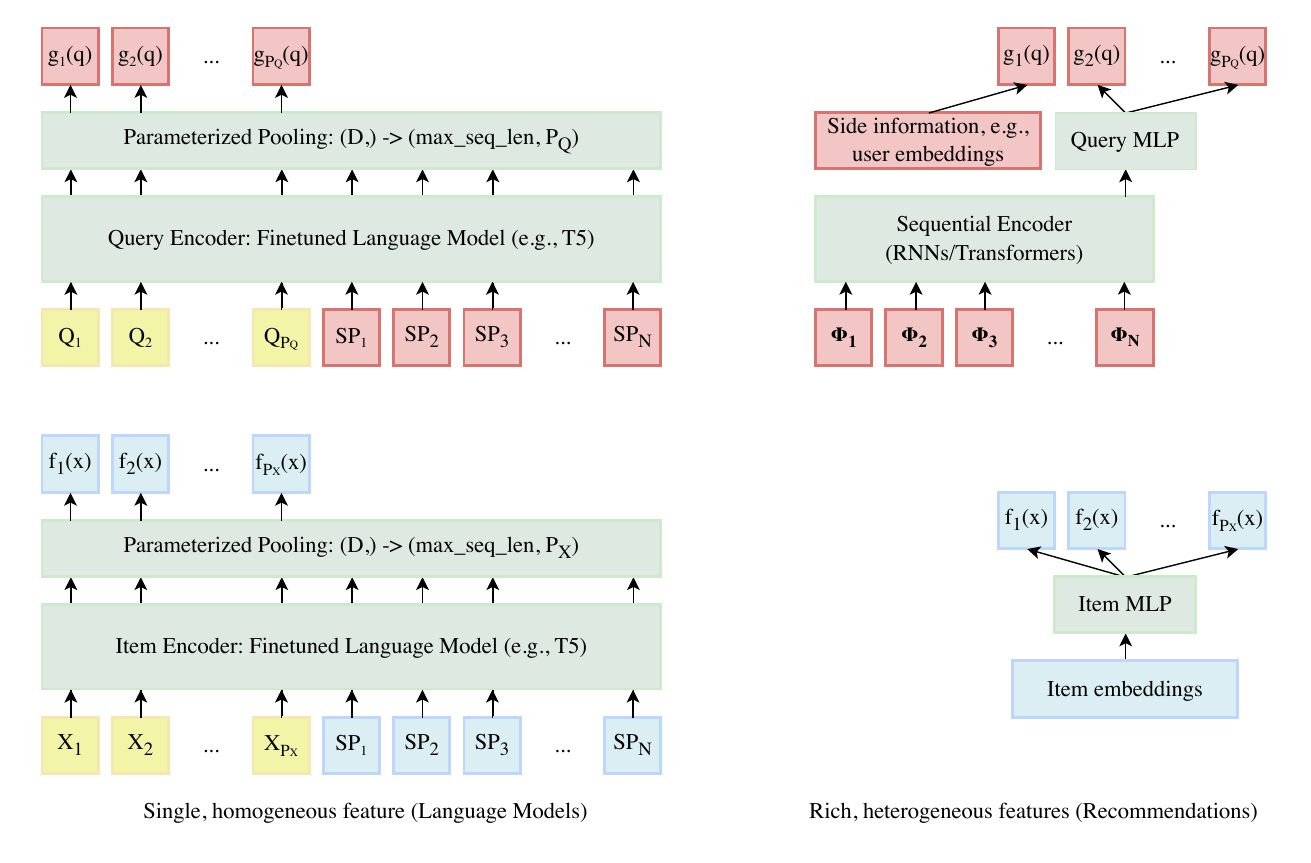}
    \vspace{-1.5em}
    \caption{Illustration of how to parameterize the embeddings to adapt Mixture-of-logits (MoL) learned similarity to various retrieval scenarios, with a language model (LM) finetuning use case in question answering (characterized by a single homogeneous feature) shown on the left, and a recommendation systems use case (characterized by a large number of heterogeneous features) shown on the right.
    For the Question Answering example on the left, $SP_1, \ldots, SP_N$ represents the original SentencePiece~\cite{sentencepiece_emnlp18} tokens that are inputs to the pre-trained language model LM, e.g., T5~\cite{t5_raffel2023exploringlimitstransferlearning}. $Q_1, Q_2, \ldots, Q_{P_Q}$ and $X_1, X_2, \ldots, X_{P_X}$ represent the special aggregation tokens we add to the LM tokenizer for pooling information across the sequence. The ``Parameterized Pooling'' component uses a $D$-dimensional embedding as input to \emph{parameterize, at an example-level}, how to weight each of the (max\_seq\_len) encoder outputs for the $P_Q$/$P_X$ MoL component-level embeddings.}
    \label{fig:mol-embedding-construction}
    \vspace{-1em}
\end{figure*}

%% file: table_multiplier.tex
\begin{table*}[t]
\centering
\caption{\todo{Impact of multipliers on top $K$ of \amzn}}
\label{table:exp:multiplier}
\begin{tabular}{|l|llll|llll|llll|llll|}
\hline
k & \multicolumn{8}{c|}{50} & \multicolumn{8}{c|}{100} \\
\hline
metric & \multicolumn{4}{c|}{recall} & \multicolumn{4}{c|}{\% of scored items} & \multicolumn{4}{c|}{recall} & \multicolumn{4}{c|}{\% of scored items} \\
\hline
multiplier & 1 & 5 & 10 & 20 & 1 & 5 & 10 & 20 & 1 & 5 & 10 & 20 & 1 & 5 & 10 & 20    \\
\hline
\topkunion          &  .00 & .00  & .00 & .00 & 0.00   &  0.00 &  0.00 &  0.00&  .00 & .00  & .00 & .00 & 0.00   &  0.00 &  0.00 &  0.00\\
\hline
\topkmax          &  .00 & .00  & .00 & .00 & 0.00   &  0.00 &  0.00 &  0.00&  .00 & .00  & .00 & .00 & 0.00   &  0.00 &  0.00 &  0.00\\
\hline
\topkavg          &  .00 & .00  & .00 & .00 & 0.00   &  0.00 &  0.00 &  0.00&  .00 & .00  & .00 & .00 & 0.00   &  0.00 &  0.00 &  0.00\\
\hline
\topkavgunion          &  .00 & .00  & .00 & .00 & 0.00   &  0.00 &  0.00 &  0.00&  .00 & .00  & .00 & .00 & 0.00   &  0.00 &  0.00 &  0.00\\
\hline
\topkavgmax          &  .00 & .00  & .00 & .00 & 0.00   &  0.00 &  0.00 &  0.00&  .00 & .00  & .00 & .00 & 0.00   &  0.00 &  0.00 &  0.00\\
\hline
\end{tabular}
\end{table*}

%% file: table_comb_ratio.tex
\begin{table*}[t]
\centering
\caption{\todo{Impact of combination ratio on top $K$ of \amzn}}
\label{table:exp:comb_ratio}
\begin{tabular}{|l|lll|lll|lll|lll|}
\hline
k & \multicolumn{6}{c|}{50} & \multicolumn{6}{c|}{100} \\
\hline
metric & \multicolumn{3}{c|}{recall} & \multicolumn{3}{c|}{\% of scored items} & \multicolumn{3}{c|}{recall} & \multicolumn{3}{c|}{\% of scored items} \\
\hline
multiplier & 1 & 5 & 10 & 1 & 5 & 10 & 1 & 5 & 10 & 1 & 5 & 10    \\
\hline
\topkunion          &  .00 & .00  & .00  &  0.00 &  0.00 &  0.00&  .00  & .00 & .00  &  0.00 &  0.00 &  0.00\\
\hline
\topkmax          &  .00 & .00  & .00  &  0.00 &  0.00 &  0.00&  .00  & .00 & .00  &  0.00 &  0.00 &  0.00\\
\hline
\topkavg          &  .00 & .00  & .00  &  0.00 &  0.00 &  0.00&  .00  & .00 & .00  &  0.00 &  0.00 &  0.00\\
\hline
\topkavgunion          &  .00 & .00  & .00  &  0.00 &  0.00 &  0.00&  .00  & .00 & .00  &  0.00 &  0.00 &  0.00\\
\hline
\topkavgmax          &  .00 & .00  & .00  &  0.00 &  0.00 &  0.00&  .00  & .00 & .00  &  0.00 &  0.00 &  0.00\\
\hline
\topkavguniontwoeight          &  .00 & .00  & .00  &  0.00 &  0.00 &  0.00&  .00  & .00 & .00  &  0.00 &  0.00 &  0.00\\
\hline
\topkavgmaxtwoeight          &  .00 & .00  & .00  &  0.00 &  0.00 &  0.00&  .00  & .00 & .00  &  0.00 &  0.00 &  0.00\\
\hline
\topkavgunioneighttwo          &  .00 & .00  & .00  &  0.00 &  0.00 &  0.00&  .00  & .00 & .00  &  0.00 &  0.00 &  0.00\\
\hline
\topkavgmaxeighttwo          &  .00 & .00  & .00  &  0.00 &  0.00 &  0.00&  .00  & .00 & .00  &  0.00 &  0.00 &  0.00\\
\hline
\end{tabular}
\end{table*}

%% file: main_www2025.bbl

\begin{thebibliography}{64}


\ifx \showCODEN    \undefined \def \showCODEN     #1{\unskip}     \fi
\ifx \showDOI      \undefined \def \showDOI       #1{#1}\fi
\ifx \showISBNx    \undefined \def \showISBNx     #1{\unskip}     \fi
\ifx \showISBNxiii \undefined \def \showISBNxiii  #1{\unskip}     \fi
\ifx \showISSN     \undefined \def \showISSN      #1{\unskip}     \fi
\ifx \showLCCN     \undefined \def \showLCCN      #1{\unskip}     \fi
\ifx \shownote     \undefined \def \shownote      #1{#1}          \fi
\ifx \showarticletitle \undefined \def \showarticletitle #1{#1}   \fi
\ifx \showURL      \undefined \def \showURL       {\relax}        \fi
\providecommand\bibfield[2]{#2}
\providecommand\bibinfo[2]{#2}
\providecommand\natexlab[1]{#1}
\providecommand\showeprint[2][]{arXiv:#2}

\bibitem[ann({[n.\,d.]})]%
        {annbenchmarks-website}
 \bibinfo{year}{[n.\,d.]}\natexlab{}.
\newblock \bibinfo{title}{{ANN} Benchmarks}.
\newblock \bibinfo{howpublished}{\url{https://ann-benchmarks.com/}}.
\newblock
\newblock
\shownote{Accessed: 2024-08-06}.


\bibitem[Bengio et~al\mbox{.}(2016)]%
        {bengio2016conditionalcomputationneuralnetworks}
\bibfield{author}{\bibinfo{person}{Emmanuel Bengio}, \bibinfo{person}{Pierre-Luc Bacon}, \bibinfo{person}{Joelle Pineau}, {and} \bibinfo{person}{Doina Precup}.} \bibinfo{year}{2016}\natexlab{}.
\newblock \bibinfo{title}{Conditional Computation in Neural Networks for faster models}.
\newblock
\newblock
\showeprint[arxiv]{1511.06297}~[cs.LG]
\urldef\tempurl%
\url{https://arxiv.org/abs/1511.06297}
\showURL{%
\tempurl}


\bibitem[Bentley(1975)]%
        {kdtree_1975}
\bibfield{author}{\bibinfo{person}{Jon~Louis Bentley}.} \bibinfo{year}{1975}\natexlab{}.
\newblock \showarticletitle{Multidimensional binary search trees used for associative searching}.
\newblock \bibinfo{journal}{\emph{Commun. ACM}} \bibinfo{volume}{18}, \bibinfo{number}{9} (\bibinfo{date}{sep} \bibinfo{year}{1975}), \bibinfo{pages}{509–517}.
\newblock
\showISSN{0001-0782}
\urldef\tempurl%
\url{https://doi.org/10.1145/361002.361007}
\showDOI{\tempurl}


\bibitem[Bevilacqua et~al\mbox{.}(2022)]%
        {seal_neurips22}
\bibfield{author}{\bibinfo{person}{Michele Bevilacqua}, \bibinfo{person}{Giuseppe Ottaviano}, \bibinfo{person}{Patrick Lewis}, \bibinfo{person}{Scott Yih}, \bibinfo{person}{Sebastian Riedel}, {and} \bibinfo{person}{Fabio Petroni}.} \bibinfo{year}{2022}\natexlab{}.
\newblock \showarticletitle{Autoregressive Search Engines: Generating Substrings as Document Identifiers}. In \bibinfo{booktitle}{\emph{Advances in Neural Information Processing Systems}}, \bibfield{editor}{\bibinfo{person}{S.~Koyejo}, \bibinfo{person}{S.~Mohamed}, \bibinfo{person}{A.~Agarwal}, \bibinfo{person}{D.~Belgrave}, \bibinfo{person}{K.~Cho}, {and} \bibinfo{person}{A.~Oh}} (Eds.), Vol.~\bibinfo{volume}{35}. \bibinfo{publisher}{Curran Associates, Inc.}, \bibinfo{pages}{31668--31683}.
\newblock
\urldef\tempurl%
\url{https://proceedings.neurips.cc/paper_files/paper/2022/file/cd88d62a2063fdaf7ce6f9068fb15dcd-Paper-Conference.pdf}
\showURL{%
\tempurl}


\bibitem[Borgeaud et~al\mbox{.}(2022)]%
        {retro_dm_icml22}
\bibfield{author}{\bibinfo{person}{Sebastian Borgeaud}, \bibinfo{person}{Arthur Mensch}, \bibinfo{person}{Jordan Hoffmann}, \bibinfo{person}{Trevor Cai}, \bibinfo{person}{Eliza Rutherford}, \bibinfo{person}{Katie Millican}, \bibinfo{person}{George van~den Driessche}, \bibinfo{person}{Jean{-}Baptiste Lespiau}, \bibinfo{person}{Bogdan Damoc}, \bibinfo{person}{Aidan Clark}, \bibinfo{person}{Diego de Las~Casas}, \bibinfo{person}{Aurelia Guy}, \bibinfo{person}{Jacob Menick}, \bibinfo{person}{Roman Ring}, \bibinfo{person}{Tom Hennigan}, \bibinfo{person}{Saffron Huang}, \bibinfo{person}{Loren Maggiore}, \bibinfo{person}{Chris Jones}, \bibinfo{person}{Albin Cassirer}, \bibinfo{person}{Andy Brock}, \bibinfo{person}{Michela Paganini}, \bibinfo{person}{Geoffrey Irving}, \bibinfo{person}{Oriol Vinyals}, \bibinfo{person}{Simon Osindero}, \bibinfo{person}{Karen Simonyan}, \bibinfo{person}{Jack~W. Rae}, \bibinfo{person}{Erich Elsen}, {and} \bibinfo{person}{Laurent Sifre}.} \bibinfo{year}{2022}\natexlab{}.
\newblock \showarticletitle{Improving Language Models by Retrieving from Trillions of Tokens}. In \bibinfo{booktitle}{\emph{International Conference on Machine Learning, {ICML} 2022, 17-23 July 2022, Baltimore, Maryland, {USA}}} \emph{(\bibinfo{series}{Proceedings of Machine Learning Research}, Vol.~\bibinfo{volume}{162})}, \bibfield{editor}{\bibinfo{person}{Kamalika Chaudhuri}, \bibinfo{person}{Stefanie Jegelka}, \bibinfo{person}{Le~Song}, \bibinfo{person}{Csaba Szepesv{\'{a}}ri}, \bibinfo{person}{Gang Niu}, {and} \bibinfo{person}{Sivan Sabato}} (Eds.). \bibinfo{publisher}{{PMLR}}, \bibinfo{pages}{2206--2240}.
\newblock
\urldef\tempurl%
\url{https://proceedings.mlr.press/v162/borgeaud22a.html}
\showURL{%
\tempurl}


\bibitem[Borisyuk et~al\mbox{.}(2024)]%
        {borisyuk2024linr_cikm24}
\bibfield{author}{\bibinfo{person}{Fedor Borisyuk}, \bibinfo{person}{Qingquan Song}, \bibinfo{person}{Mingzhou Zhou}, \bibinfo{person}{Ganesh Parameswaran}, \bibinfo{person}{Madhu Arun}, \bibinfo{person}{Siva Popuri}, \bibinfo{person}{Tugrul Bingol}, \bibinfo{person}{Zhuotao Pei}, \bibinfo{person}{Kuang-Hsuan Lee}, \bibinfo{person}{Lu Zheng}, \bibinfo{person}{Qizhan Shao}, \bibinfo{person}{Ali Naqvi}, \bibinfo{person}{Sen Zhou}, {and} \bibinfo{person}{Aman Gupta}.} \bibinfo{year}{2024}\natexlab{}.
\newblock \showarticletitle{LiNR: Model Based Neural Retrieval on GPUs at LinkedIn}. In \bibinfo{booktitle}{\emph{Proceedings of the 33rd ACM International Conference on Information and Knowledge Management}} (Boise, ID, USA) \emph{(\bibinfo{series}{CIKM '24})}. \bibinfo{publisher}{Association for Computing Machinery}, \bibinfo{address}{New York, NY, USA}, \bibinfo{pages}{4366–4373}.
\newblock
\showISBNx{9798400704369}
\urldef\tempurl%
\url{https://doi.org/10.1145/3627673.3680091}
\showDOI{\tempurl}


\bibitem[Chang et~al\mbox{.}(2023)]%
        {chang-etal-2023-multi_acl23}
\bibfield{author}{\bibinfo{person}{Haw-Shiuan Chang}, \bibinfo{person}{Ruei-Yao Sun}, \bibinfo{person}{Kathryn Ricci}, {and} \bibinfo{person}{Andrew McCallum}.} \bibinfo{year}{2023}\natexlab{}.
\newblock \showarticletitle{Multi-{CLS} {BERT}: An Efficient Alternative to Traditional Ensembling}. In \bibinfo{booktitle}{\emph{Proceedings of the 61st Annual Meeting of the Association for Computational Linguistics (Volume 1: Long Papers)}}, \bibfield{editor}{\bibinfo{person}{Anna Rogers}, \bibinfo{person}{Jordan Boyd-Graber}, {and} \bibinfo{person}{Naoaki Okazaki}} (Eds.). \bibinfo{publisher}{Association for Computational Linguistics}, \bibinfo{address}{Toronto, Canada}, \bibinfo{pages}{821--854}.
\newblock
\urldef\tempurl%
\url{https://doi.org/10.18653/v1/2023.acl-long.48}
\showDOI{\tempurl}


\bibitem[Chen et~al\mbox{.}(2023)]%
        {modsquad_cvpr23}
\bibfield{author}{\bibinfo{person}{Zitian Chen}, \bibinfo{person}{Yikang Shen}, \bibinfo{person}{Mingyu Ding}, \bibinfo{person}{Zhenfang Chen}, \bibinfo{person}{Hengshuang Zhao}, \bibinfo{person}{Erik~G. Learned-Miller}, {and} \bibinfo{person}{Chuang Gan}.} \bibinfo{year}{2023}\natexlab{}.
\newblock \showarticletitle{Mod-Squad: Designing Mixtures of Experts As Modular Multi-Task Learners}. In \bibinfo{booktitle}{\emph{Proceedings of the IEEE/CVF Conference on Computer Vision and Pattern Recognition (CVPR)}}. \bibinfo{pages}{11828--11837}.
\newblock


\bibitem[Chern et~al\mbox{.}(2022)]%
        {tpuknn_goog_neurips22}
\bibfield{author}{\bibinfo{person}{Felix Chern}, \bibinfo{person}{Blake Hechtman}, \bibinfo{person}{Andy Davis}, \bibinfo{person}{Ruiqi Guo}, \bibinfo{person}{David Majnemer}, {and} \bibinfo{person}{Sanjiv Kumar}.} \bibinfo{year}{2022}\natexlab{}.
\newblock \showarticletitle{{TPU}-{KNN}: K Nearest Neighbor Search at Peak {FLOP}/s}. In \bibinfo{booktitle}{\emph{Advances in Neural Information Processing Systems}}.
\newblock


\bibitem[Covington et~al\mbox{.}(2016)]%
        {ytdnn_goog_recsys16}
\bibfield{author}{\bibinfo{person}{Paul Covington}, \bibinfo{person}{Jay Adams}, {and} \bibinfo{person}{Emre Sargin}.} \bibinfo{year}{2016}\natexlab{}.
\newblock \showarticletitle{Deep Neural Networks for YouTube Recommendations}. In \bibinfo{booktitle}{\emph{Proceedings of the 10th ACM Conference on Recommender Systems}} \emph{(\bibinfo{series}{RecSys '16})}. \bibinfo{pages}{191–198}.
\newblock
\showISBNx{9781450340359}


\bibitem[{De Cao} et~al\mbox{.}(2021)]%
        {genre_iclr21}
\bibfield{author}{\bibinfo{person}{Nicola {De Cao}}, \bibinfo{person}{Gautier Izacard}, \bibinfo{person}{Sebastian Riedel}, {and} \bibinfo{person}{Fabio Petroni}.} \bibinfo{year}{2021}\natexlab{}.
\newblock \showarticletitle{Autoregressive Entity Retrieval}. In \bibinfo{booktitle}{\emph{9th International Conference on Learning Representations, {ICLR} 2021, Virtual Event, Austria, May 3-7, 2021}}. \bibinfo{publisher}{OpenReview.net}.
\newblock
\urldef\tempurl%
\url{https://openreview.net/forum?id=5k8F6UU39V}
\showURL{%
\tempurl}


\bibitem[Devlin et~al\mbox{.}(2019)]%
        {bert_naacl19}
\bibfield{author}{\bibinfo{person}{Jacob Devlin}, \bibinfo{person}{Ming{-}Wei Chang}, \bibinfo{person}{Kenton Lee}, {and} \bibinfo{person}{Kristina Toutanova}.} \bibinfo{year}{2019}\natexlab{}.
\newblock \showarticletitle{{BERT:} Pre-training of Deep Bidirectional Transformers for Language Understanding}. In \bibinfo{booktitle}{\emph{Proceedings of the 2019 Conference of the North American Chapter of the Association for Computational Linguistics: Human Language Technologies, {NAACL-HLT} 2019, Minneapolis, MN, USA, June 2-7, 2019, Volume 1 (Long and Short Papers)}}, \bibfield{editor}{\bibinfo{person}{Jill Burstein}, \bibinfo{person}{Christy Doran}, {and} \bibinfo{person}{Thamar Solorio}} (Eds.). \bibinfo{publisher}{Association for Computational Linguistics}, \bibinfo{pages}{4171--4186}.
\newblock
\urldef\tempurl%
\url{https://doi.org/10.18653/v1/n19-1423}
\showDOI{\tempurl}


\bibitem[Eksombatchai et~al\mbox{.}(2018)]%
        {pixie_pins_www18}
\bibfield{author}{\bibinfo{person}{Chantat Eksombatchai}, \bibinfo{person}{Pranav Jindal}, \bibinfo{person}{Jerry~Zitao Liu}, \bibinfo{person}{Yuchen Liu}, \bibinfo{person}{Rahul Sharma}, \bibinfo{person}{Charles Sugnet}, \bibinfo{person}{Mark Ulrich}, {and} \bibinfo{person}{Jure Leskovec}.} \bibinfo{year}{2018}\natexlab{}.
\newblock \showarticletitle{Pixie: A System for Recommending 3+ Billion Items to 200+ Million Users in Real-Time}. In \bibinfo{booktitle}{\emph{Proceedings of the 2018 World Wide Web Conference}} \emph{(\bibinfo{series}{WWW '18})}. \bibinfo{pages}{1775–1784}.
\newblock
\showISBNx{9781450356398}


\bibitem[Elfwing et~al\mbox{.}(2017)]%
        {elfwing2017silu}
\bibfield{author}{\bibinfo{person}{Stefan Elfwing}, \bibinfo{person}{Eiji Uchibe}, {and} \bibinfo{person}{Kenji Doya}.} \bibinfo{year}{2017}\natexlab{}.
\newblock \showarticletitle{Sigmoid-Weighted Linear Units for Neural Network Function Approximation in Reinforcement Learning}.
\newblock \bibinfo{journal}{\emph{CoRR}}  \bibinfo{volume}{abs/1702.03118} (\bibinfo{year}{2017}).
\newblock
\showeprint[arXiv]{1702.03118}
\urldef\tempurl%
\url{http://arxiv.org/abs/1702.03118}
\showURL{%
\tempurl}


\bibitem[Gao et~al\mbox{.}(2021)]%
        {dr_cikm21}
\bibfield{author}{\bibinfo{person}{Weihao Gao}, \bibinfo{person}{Xiangjun Fan}, \bibinfo{person}{Chong Wang}, \bibinfo{person}{Jiankai Sun}, \bibinfo{person}{Kai Jia}, \bibinfo{person}{Wenzi Xiao}, \bibinfo{person}{Ruofan Ding}, \bibinfo{person}{Xingyan Bin}, \bibinfo{person}{Hui Yang}, {and} \bibinfo{person}{Xiaobing Liu}.} \bibinfo{year}{2021}\natexlab{}.
\newblock \showarticletitle{Learning An End-to-End Structure for Retrieval in Large-Scale Recommendations}. In \bibinfo{booktitle}{\emph{Proceedings of the 30th ACM International Conference on Information and Knowledge Management}} \emph{(\bibinfo{series}{CIKM '21})}. \bibinfo{pages}{524–533}.
\newblock
\showISBNx{9781450384469}


\bibitem[Gillick et~al\mbox{.}(2018)]%
        {gillick2018endtoend}
\bibfield{author}{\bibinfo{person}{Daniel Gillick}, \bibinfo{person}{Alessandro Presta}, {and} \bibinfo{person}{Gaurav~Singh Tomar}.} \bibinfo{year}{2018}\natexlab{}.
\newblock \bibinfo{title}{End-to-End Retrieval in Continuous Space}.
\newblock
\newblock
\showeprint[arxiv]{1811.08008}~[cs.IR]


\bibitem[Gionis et~al\mbox{.}(1999)]%
        {lsh_vldb99}
\bibfield{author}{\bibinfo{person}{Aristides Gionis}, \bibinfo{person}{Piotr Indyk}, {and} \bibinfo{person}{Rajeev Motwani}.} \bibinfo{year}{1999}\natexlab{}.
\newblock \showarticletitle{Similarity Search in High Dimensions via Hashing}. In \bibinfo{booktitle}{\emph{Proceedings of the 25th International Conference on Very Large Data Bases}} \emph{(\bibinfo{series}{VLDB '99})}. \bibinfo{publisher}{Morgan Kaufmann Publishers Inc.}, \bibinfo{address}{San Francisco, CA, USA}, \bibinfo{pages}{518–529}.
\newblock
\showISBNx{1558606157}


\bibitem[Guo et~al\mbox{.}(2016)]%
        {quant_mips_kumar_aistats16}
\bibfield{author}{\bibinfo{person}{Ruiqi Guo}, \bibinfo{person}{Sanjiv Kumar}, \bibinfo{person}{Krzysztof Choromanski}, {and} \bibinfo{person}{David Simcha}.} \bibinfo{year}{2016}\natexlab{}.
\newblock \showarticletitle{Quantization based Fast Inner Product Search}. In \bibinfo{booktitle}{\emph{Proceedings of the 19th International Conference on Artificial Intelligence and Statistics, {AISTATS} 2016}}, Vol.~\bibinfo{volume}{51}. \bibinfo{pages}{482--490}.
\newblock


\bibitem[Guo et~al\mbox{.}(2020)]%
        {scann_icml20}
\bibfield{author}{\bibinfo{person}{Ruiqi Guo}, \bibinfo{person}{Philip Sun}, \bibinfo{person}{Erik Lindgren}, \bibinfo{person}{Quan Geng}, \bibinfo{person}{David Simcha}, \bibinfo{person}{Felix Chern}, {and} \bibinfo{person}{Sanjiv Kumar}.} \bibinfo{year}{2020}\natexlab{}.
\newblock \showarticletitle{Accelerating large-scale inference with anisotropic vector quantization}. In \bibinfo{booktitle}{\emph{Proceedings of the 37th International Conference on Machine Learning}} \emph{(\bibinfo{series}{ICML'20})}. \bibinfo{publisher}{JMLR.org}, Article \bibinfo{articleno}{364}, \bibinfo{numpages}{10}~pages.
\newblock


\bibitem[Harper and Konstan(2015)]%
        {movielens_2015}
\bibfield{author}{\bibinfo{person}{F.~Maxwell Harper} {and} \bibinfo{person}{Joseph~A. Konstan}.} \bibinfo{year}{2015}\natexlab{}.
\newblock \showarticletitle{The MovieLens Datasets: History and Context}.
\newblock \bibinfo{journal}{\emph{ACM Trans. Interact. Intell. Syst.}} \bibinfo{volume}{5}, \bibinfo{number}{4}, Article \bibinfo{articleno}{19} (\bibinfo{date}{dec} \bibinfo{year}{2015}), \bibinfo{numpages}{19}~pages.
\newblock
\showISSN{2160-6455}
\urldef\tempurl%
\url{https://doi.org/10.1145/2827872}
\showDOI{\tempurl}


\bibitem[He et~al\mbox{.}(2017)]%
        {ncf_www17}
\bibfield{author}{\bibinfo{person}{Xiangnan He}, \bibinfo{person}{Lizi Liao}, \bibinfo{person}{Hanwang Zhang}, \bibinfo{person}{Liqiang Nie}, \bibinfo{person}{Xia Hu}, {and} \bibinfo{person}{Tat-Seng Chua}.} \bibinfo{year}{2017}\natexlab{}.
\newblock \showarticletitle{Neural Collaborative Filtering}. In \bibinfo{booktitle}{\emph{Proceedings of the 26th International Conference on World Wide Web}} (Perth, Australia) \emph{(\bibinfo{series}{WWW '17})}. \bibinfo{pages}{173–182}.
\newblock
\showISBNx{9781450349130}


\bibitem[Hidasi et~al\mbox{.}(2016)]%
        {gru4rec_iclr16}
\bibfield{author}{\bibinfo{person}{Bal{\'{a}}zs Hidasi}, \bibinfo{person}{Alexandros Karatzoglou}, \bibinfo{person}{Linas Baltrunas}, {and} \bibinfo{person}{Domonkos Tikk}.} \bibinfo{year}{2016}\natexlab{}.
\newblock \showarticletitle{Session-based Recommendations with Recurrent Neural Networks}. In \bibinfo{booktitle}{\emph{4th International Conference on Learning Representations, {ICLR} 2016, San Juan, Puerto Rico, May 2-4, 2016, Conference Track Proceedings}}, \bibfield{editor}{\bibinfo{person}{Yoshua Bengio} {and} \bibinfo{person}{Yann LeCun}} (Eds.).
\newblock
\urldef\tempurl%
\url{http://arxiv.org/abs/1511.06939}
\showURL{%
\tempurl}


\bibitem[Jasinska et~al\mbox{.}(2016)]%
        {plt_icml16}
\bibfield{author}{\bibinfo{person}{Kalina Jasinska}, \bibinfo{person}{Krzysztof Dembczynski}, \bibinfo{person}{Robert Busa-Fekete}, \bibinfo{person}{Karlson Pfannschmidt}, \bibinfo{person}{Timo Klerx}, {and} \bibinfo{person}{Eyke Hullermeier}.} \bibinfo{year}{2016}\natexlab{}.
\newblock \showarticletitle{Extreme F-measure Maximization using Sparse Probability Estimates}. In \bibinfo{booktitle}{\emph{Proceedings of The 33rd International Conference on Machine Learning}} \emph{(\bibinfo{series}{Proceedings of Machine Learning Research}, Vol.~\bibinfo{volume}{48})}, \bibfield{editor}{\bibinfo{person}{Maria~Florina Balcan} {and} \bibinfo{person}{Kilian~Q. Weinberger}} (Eds.). \bibinfo{publisher}{PMLR}, \bibinfo{address}{New York, New York, USA}, \bibinfo{pages}{1435--1444}.
\newblock
\urldef\tempurl%
\url{https://proceedings.mlr.press/v48/jasinska16.html}
\showURL{%
\tempurl}


\bibitem[Jayaram~Subramanya et~al\mbox{.}(2019)]%
        {diskann_neurips19}
\bibfield{author}{\bibinfo{person}{Suhas Jayaram~Subramanya}, \bibinfo{person}{Fnu Devvrit}, \bibinfo{person}{Harsha~Vardhan Simhadri}, \bibinfo{person}{Ravishankar Krishnawamy}, {and} \bibinfo{person}{Rohan Kadekodi}.} \bibinfo{year}{2019}\natexlab{}.
\newblock \showarticletitle{DiskANN: Fast Accurate Billion-point Nearest Neighbor Search on a Single Node}. In \bibinfo{booktitle}{\emph{Advances in Neural Information Processing Systems}}, \bibfield{editor}{\bibinfo{person}{H.~Wallach}, \bibinfo{person}{H.~Larochelle}, \bibinfo{person}{A.~Beygelzimer}, \bibinfo{person}{F.~d\textquotesingle Alch\'{e}-Buc}, \bibinfo{person}{E.~Fox}, {and} \bibinfo{person}{R.~Garnett}} (Eds.), Vol.~\bibinfo{volume}{32}. \bibinfo{publisher}{Curran Associates, Inc.}
\newblock
\urldef\tempurl%
\url{https://proceedings.neurips.cc/paper_files/paper/2019/file/09853c7fb1d3f8ee67a61b6bf4a7f8e6-Paper.pdf}
\showURL{%
\tempurl}


\bibitem[Jegou et~al\mbox{.}(2011)]%
        {pq_nns_pami11}
\bibfield{author}{\bibinfo{person}{Herve Jegou}, \bibinfo{person}{Matthijs Douze}, {and} \bibinfo{person}{Cordelia Schmid}.} \bibinfo{year}{2011}\natexlab{}.
\newblock \showarticletitle{Product Quantization for Nearest Neighbor Search}.
\newblock \bibinfo{journal}{\emph{IEEE Trans. Pattern Anal. Mach. Intell.}} \bibinfo{volume}{33}, \bibinfo{number}{1} (\bibinfo{date}{jan} \bibinfo{year}{2011}), \bibinfo{pages}{117–128}.
\newblock
\showISSN{0162-8828}
\urldef\tempurl%
\url{https://doi.org/10.1109/TPAMI.2010.57}
\showDOI{\tempurl}


\bibitem[Johnson et~al\mbox{.}(2021)]%
        {faiss_tbd21}
\bibfield{author}{\bibinfo{person}{J. Johnson}, \bibinfo{person}{M. Douze}, {and} \bibinfo{person}{H. Jegou}.} \bibinfo{year}{2021}\natexlab{}.
\newblock \showarticletitle{Billion-Scale Similarity Search with GPUs}.
\newblock \bibinfo{journal}{\emph{IEEE Transactions on Big Data}} \bibinfo{volume}{7}, \bibinfo{number}{03} (\bibinfo{date}{Jul} \bibinfo{year}{2021}), \bibinfo{pages}{535--547}.
\newblock
\showISSN{2332-7790}


\bibitem[Kang and McAuley(2018)]%
        {sasrec_icdm18}
\bibfield{author}{\bibinfo{person}{Wang-Cheng Kang} {and} \bibinfo{person}{Julian McAuley}.} \bibinfo{year}{2018}\natexlab{}.
\newblock \showarticletitle{Self-attentive sequential recommendation}. In \bibinfo{booktitle}{\emph{2018 International Conference on Data Mining (ICDM)}}. \bibinfo{pages}{197--206}.
\newblock


\bibitem[Karpukhin et~al\mbox{.}(2020)]%
        {dpr_emnlp20}
\bibfield{author}{\bibinfo{person}{Vladimir Karpukhin}, \bibinfo{person}{Barlas Oguz}, \bibinfo{person}{Sewon Min}, \bibinfo{person}{Patrick Lewis}, \bibinfo{person}{Ledell Wu}, \bibinfo{person}{Sergey Edunov}, \bibinfo{person}{Danqi Chen}, {and} \bibinfo{person}{Wen-tau Yih}.} \bibinfo{year}{2020}\natexlab{}.
\newblock \showarticletitle{Dense Passage Retrieval for Open-Domain Question Answering}. In \bibinfo{booktitle}{\emph{Proceedings of the 2020 Conference on Empirical Methods in Natural Language Processing (EMNLP)}}, \bibfield{editor}{\bibinfo{person}{Bonnie Webber}, \bibinfo{person}{Trevor Cohn}, \bibinfo{person}{Yulan He}, {and} \bibinfo{person}{Yang Liu}} (Eds.). \bibinfo{publisher}{Association for Computational Linguistics}, \bibinfo{address}{Online}, \bibinfo{pages}{6769--6781}.
\newblock
\urldef\tempurl%
\url{https://doi.org/10.18653/v1/2020.emnlp-main.550}
\showDOI{\tempurl}


\bibitem[Khattab and Zaharia(2020)]%
        {colbert_sigir20}
\bibfield{author}{\bibinfo{person}{Omar Khattab} {and} \bibinfo{person}{Matei Zaharia}.} \bibinfo{year}{2020}\natexlab{}.
\newblock \showarticletitle{ColBERT: Efficient and Effective Passage Search via Contextualized Late Interaction over BERT}. In \bibinfo{booktitle}{\emph{Proceedings of the 43rd International ACM SIGIR Conference on Research and Development in Information Retrieval}} (Virtual Event, China) \emph{(\bibinfo{series}{SIGIR '20})}. \bibinfo{publisher}{Association for Computing Machinery}, \bibinfo{address}{New York, NY, USA}, \bibinfo{pages}{39–48}.
\newblock
\showISBNx{9781450380164}
\urldef\tempurl%
\url{https://doi.org/10.1145/3397271.3401075}
\showDOI{\tempurl}


\bibitem[Koren et~al\mbox{.}(2009)]%
        {matrix-factorization-netflix09}
\bibfield{author}{\bibinfo{person}{Yehuda Koren}, \bibinfo{person}{Robert Bell}, {and} \bibinfo{person}{Chris Volinsky}.} \bibinfo{year}{2009}\natexlab{}.
\newblock \showarticletitle{Matrix Factorization Techniques for Recommender Systems}.
\newblock \bibinfo{journal}{\emph{Computer}} \bibinfo{volume}{42}, \bibinfo{number}{8} (\bibinfo{year}{2009}), \bibinfo{pages}{30--37}.
\newblock
\urldef\tempurl%
\url{https://doi.org/10.1109/MC.2009.263}
\showDOI{\tempurl}


\bibitem[Kudo and Richardson(2018)]%
        {sentencepiece_emnlp18}
\bibfield{author}{\bibinfo{person}{Taku Kudo} {and} \bibinfo{person}{John Richardson}.} \bibinfo{year}{2018}\natexlab{}.
\newblock \showarticletitle{{S}entence{P}iece: A simple and language independent subword tokenizer and detokenizer for Neural Text Processing}. In \bibinfo{booktitle}{\emph{Proceedings of the 2018 Conference on Empirical Methods in Natural Language Processing: System Demonstrations}}, \bibfield{editor}{\bibinfo{person}{Eduardo Blanco} {and} \bibinfo{person}{Wei Lu}} (Eds.). \bibinfo{publisher}{Association for Computational Linguistics}, \bibinfo{address}{Brussels, Belgium}, \bibinfo{pages}{66--71}.
\newblock
\urldef\tempurl%
\url{https://doi.org/10.18653/v1/D18-2012}
\showDOI{\tempurl}


\bibitem[Kwiatkowski et~al\mbox{.}(2019)]%
        {naturalquestions_tacl19}
\bibfield{author}{\bibinfo{person}{Tom Kwiatkowski}, \bibinfo{person}{Jennimaria Palomaki}, \bibinfo{person}{Olivia Redfield}, \bibinfo{person}{Michael Collins}, \bibinfo{person}{Ankur Parikh}, \bibinfo{person}{Chris Alberti}, \bibinfo{person}{Danielle Epstein}, \bibinfo{person}{Illia Polosukhin}, \bibinfo{person}{Jacob Devlin}, \bibinfo{person}{Kenton Lee}, \bibinfo{person}{Kristina Toutanova}, \bibinfo{person}{Llion Jones}, \bibinfo{person}{Matthew Kelcey}, \bibinfo{person}{Ming-Wei Chang}, \bibinfo{person}{Andrew~M. Dai}, \bibinfo{person}{Jakob Uszkoreit}, \bibinfo{person}{Quoc Le}, {and} \bibinfo{person}{Slav Petrov}.} \bibinfo{year}{2019}\natexlab{}.
\newblock \showarticletitle{Natural Questions: A Benchmark for Question Answering Research}.
\newblock \bibinfo{journal}{\emph{Transactions of the Association for Computational Linguistics}}  \bibinfo{volume}{7} (\bibinfo{year}{2019}), \bibinfo{pages}{452--466}.
\newblock
\urldef\tempurl%
\url{https://doi.org/10.1162/tacl_a_00276}
\showDOI{\tempurl}


\bibitem[Lewis et~al\mbox{.}(2020)]%
        {rag_neurips20}
\bibfield{author}{\bibinfo{person}{Patrick Lewis}, \bibinfo{person}{Ethan Perez}, \bibinfo{person}{Aleksandra Piktus}, \bibinfo{person}{Fabio Petroni}, \bibinfo{person}{Vladimir Karpukhin}, \bibinfo{person}{Naman Goyal}, \bibinfo{person}{Heinrich K\"{u}ttler}, \bibinfo{person}{Mike Lewis}, \bibinfo{person}{Wen-tau Yih}, \bibinfo{person}{Tim Rockt\"{a}schel}, \bibinfo{person}{Sebastian Riedel}, {and} \bibinfo{person}{Douwe Kiela}.} \bibinfo{year}{2020}\natexlab{}.
\newblock \showarticletitle{Retrieval-augmented generation for knowledge-intensive NLP tasks}. In \bibinfo{booktitle}{\emph{Proceedings of the 34th International Conference on Neural Information Processing Systems}} (, Vancouver, BC, Canada,) \emph{(\bibinfo{series}{NIPS '20})}. \bibinfo{publisher}{Curran Associates Inc.}, \bibinfo{address}{Red Hook, NY, USA}, Article \bibinfo{articleno}{793}, \bibinfo{numpages}{16}~pages.
\newblock
\showISBNx{9781713829546}


\bibitem[Li et~al\mbox{.}(2002)]%
        {clusteringss_tkde02}
\bibfield{author}{\bibinfo{person}{Chen Li}, \bibinfo{person}{E. Chang}, \bibinfo{person}{H. Garcia-Molina}, {and} \bibinfo{person}{G. Wiederhold}.} \bibinfo{year}{2002}\natexlab{}.
\newblock \showarticletitle{Clustering for approximate similarity search in high-dimensional spaces}.
\newblock \bibinfo{journal}{\emph{IEEE Transactions on Knowledge and Data Engineering}} \bibinfo{volume}{14}, \bibinfo{number}{4} (\bibinfo{year}{2002}), \bibinfo{pages}{792--808}.
\newblock


\bibitem[Li et~al\mbox{.}(2019)]%
        {mind_baba_cikm19}
\bibfield{author}{\bibinfo{person}{Chao Li}, \bibinfo{person}{Zhiyuan Liu}, \bibinfo{person}{Mengmeng Wu}, \bibinfo{person}{Yuchi Xu}, \bibinfo{person}{Huan Zhao}, \bibinfo{person}{Pipei Huang}, \bibinfo{person}{Guoliang Kang}, \bibinfo{person}{Qiwei Chen}, \bibinfo{person}{Wei Li}, {and} \bibinfo{person}{Dik~Lun Lee}.} \bibinfo{year}{2019}\natexlab{}.
\newblock \showarticletitle{Multi-Interest Network with Dynamic Routing for Recommendation at Tmall}. In \bibinfo{booktitle}{\emph{Proceedings of the 28th ACM International Conference on Information and Knowledge Management}} \emph{(\bibinfo{series}{CIKM '19})}. \bibinfo{pages}{2615–2623}.
\newblock
\showISBNx{9781450369763}


\bibitem[Loshchilov and Hutter(2019)]%
        {loshchilov2018decoupled_iclr19}
\bibfield{author}{\bibinfo{person}{Ilya Loshchilov} {and} \bibinfo{person}{Frank Hutter}.} \bibinfo{year}{2019}\natexlab{}.
\newblock \showarticletitle{Decoupled Weight Decay Regularization}. In \bibinfo{booktitle}{\emph{International Conference on Learning Representations}}.
\newblock
\urldef\tempurl%
\url{https://openreview.net/forum?id=Bkg6RiCqY7}
\showURL{%
\tempurl}


\bibitem[Malkov and Yashunin(2020)]%
        {hnsw_tpami18}
\bibfield{author}{\bibinfo{person}{Yu~A. Malkov} {and} \bibinfo{person}{D.~A. Yashunin}.} \bibinfo{year}{2020}\natexlab{}.
\newblock \showarticletitle{Efficient and Robust Approximate Nearest Neighbor Search Using Hierarchical Navigable Small World Graphs}.
\newblock \bibinfo{journal}{\emph{IEEE Trans. Pattern Anal. Mach. Intell.}} \bibinfo{volume}{42}, \bibinfo{number}{4} (\bibinfo{date}{apr} \bibinfo{year}{2020}), \bibinfo{pages}{824–836}.
\newblock
\showISSN{0162-8828}
\urldef\tempurl%
\url{https://doi.org/10.1109/TPAMI.2018.2889473}
\showDOI{\tempurl}


\bibitem[McAuley et~al\mbox{.}(2015)]%
        {amznreviews_sigir15}
\bibfield{author}{\bibinfo{person}{Julian McAuley}, \bibinfo{person}{Christopher Targett}, \bibinfo{person}{Qinfeng Shi}, {and} \bibinfo{person}{Anton van~den Hengel}.} \bibinfo{year}{2015}\natexlab{}.
\newblock \showarticletitle{Image-Based Recommendations on Styles and Substitutes}. In \bibinfo{booktitle}{\emph{Proceedings of the 38th International ACM SIGIR Conference on Research and Development in Information Retrieval}} (Santiago, Chile) \emph{(\bibinfo{series}{SIGIR '15})}. \bibinfo{publisher}{Association for Computing Machinery}, \bibinfo{address}{New York, NY, USA}, \bibinfo{pages}{43–52}.
\newblock
\showISBNx{9781450336215}
\urldef\tempurl%
\url{https://doi.org/10.1145/2766462.2767755}
\showDOI{\tempurl}


\bibitem[Morozov and Babenko(2018)]%
        {yandex2018neurips}
\bibfield{author}{\bibinfo{person}{Stanislav Morozov} {and} \bibinfo{person}{Artem Babenko}.} \bibinfo{year}{2018}\natexlab{}.
\newblock \showarticletitle{Non-metric Similarity Graphs for Maximum Inner Product Search}. In \bibinfo{booktitle}{\emph{Advances in Neural Information Processing Systems}}, \bibfield{editor}{\bibinfo{person}{S.~Bengio}, \bibinfo{person}{H.~Wallach}, \bibinfo{person}{H.~Larochelle}, \bibinfo{person}{K.~Grauman}, \bibinfo{person}{N.~Cesa-Bianchi}, {and} \bibinfo{person}{R.~Garnett}} (Eds.), Vol.~\bibinfo{volume}{31}. \bibinfo{publisher}{Curran Associates, Inc.}
\newblock
\urldef\tempurl%
\url{https://proceedings.neurips.cc/paper_files/paper/2018/file/229754d7799160502a143a72f6789927-Paper.pdf}
\showURL{%
\tempurl}


\bibitem[Ni et~al\mbox{.}(2022a)]%
        {sentencet5_acl2022}
\bibfield{author}{\bibinfo{person}{Jianmo Ni}, \bibinfo{person}{Gustavo Hernandez~Abrego}, \bibinfo{person}{Noah Constant}, \bibinfo{person}{Ji Ma}, \bibinfo{person}{Keith Hall}, \bibinfo{person}{Daniel Cer}, {and} \bibinfo{person}{Yinfei Yang}.} \bibinfo{year}{2022}\natexlab{a}.
\newblock \showarticletitle{Sentence-T5: Scalable Sentence Encoders from Pre-trained Text-to-Text Models}. In \bibinfo{booktitle}{\emph{Findings of the Association for Computational Linguistics: ACL 2022}}, \bibfield{editor}{\bibinfo{person}{Smaranda Muresan}, \bibinfo{person}{Preslav Nakov}, {and} \bibinfo{person}{Aline Villavicencio}} (Eds.). \bibinfo{publisher}{Association for Computational Linguistics}, \bibinfo{address}{Dublin, Ireland}, \bibinfo{pages}{1864--1874}.
\newblock
\urldef\tempurl%
\url{https://doi.org/10.18653/v1/2022.findings-acl.146}
\showDOI{\tempurl}


\bibitem[Ni et~al\mbox{.}(2022b)]%
        {gtr_emnlp22}
\bibfield{author}{\bibinfo{person}{Jianmo Ni}, \bibinfo{person}{Chen Qu}, \bibinfo{person}{Jing Lu}, \bibinfo{person}{Zhuyun Dai}, \bibinfo{person}{Gustavo Hernandez~Abrego}, \bibinfo{person}{Ji Ma}, \bibinfo{person}{Vincent Zhao}, \bibinfo{person}{Yi Luan}, \bibinfo{person}{Keith Hall}, \bibinfo{person}{Ming-Wei Chang}, {and} \bibinfo{person}{Yinfei Yang}.} \bibinfo{year}{2022}\natexlab{b}.
\newblock \showarticletitle{Large Dual Encoders Are Generalizable Retrievers}. In \bibinfo{booktitle}{\emph{Proceedings of the 2022 Conference on Empirical Methods in Natural Language Processing}}, \bibfield{editor}{\bibinfo{person}{Yoav Goldberg}, \bibinfo{person}{Zornitsa Kozareva}, {and} \bibinfo{person}{Yue Zhang}} (Eds.). \bibinfo{publisher}{Association for Computational Linguistics}, \bibinfo{address}{Abu Dhabi, United Arab Emirates}, \bibinfo{pages}{9844--9855}.
\newblock
\urldef\tempurl%
\url{https://doi.org/10.18653/v1/2022.emnlp-main.669}
\showDOI{\tempurl}


\bibitem[Nogueira and Lin(2019)]%
        {dott5query_2019}
\bibfield{author}{\bibinfo{person}{Rodrigo Nogueira} {and} \bibinfo{person}{Jimmy Lin}.} \bibinfo{year}{2019}\natexlab{}.
\newblock \bibinfo{title}{From doc2query to doctttttquery}.
\newblock
\newblock
\urldef\tempurl%
\url{https://cs.uwaterloo.ca/~jimmylin/publications/Nogueira_Lin_2019_docTTTTTquery-v2.pdf}
\showURL{%
\tempurl}


\bibitem[Ootomo et~al\mbox{.}(2024)]%
        {cagra_icde24}
\bibfield{author}{\bibinfo{person}{Hiroyuki Ootomo}, \bibinfo{person}{Akira Naruse}, \bibinfo{person}{Corey Nolet}, \bibinfo{person}{Ray Wang}, \bibinfo{person}{Tamas Feher}, {and} \bibinfo{person}{Yong Wang}.} \bibinfo{year}{2024}\natexlab{}.
\newblock \bibinfo{title}{CAGRA: Highly Parallel Graph Construction and Approximate Nearest Neighbor Search for GPUs}.
\newblock
\newblock


\bibitem[Raffel et~al\mbox{.}(2023)]%
        {t5_raffel2023exploringlimitstransferlearning}
\bibfield{author}{\bibinfo{person}{Colin Raffel}, \bibinfo{person}{Noam Shazeer}, \bibinfo{person}{Adam Roberts}, \bibinfo{person}{Katherine Lee}, \bibinfo{person}{Sharan Narang}, \bibinfo{person}{Michael Matena}, \bibinfo{person}{Yanqi Zhou}, \bibinfo{person}{Wei Li}, {and} \bibinfo{person}{Peter~J. Liu}.} \bibinfo{year}{2023}\natexlab{}.
\newblock \bibinfo{title}{Exploring the Limits of Transfer Learning with a Unified Text-to-Text Transformer}.
\newblock
\newblock
\showeprint[arxiv]{1910.10683}~[cs.LG]
\urldef\tempurl%
\url{https://arxiv.org/abs/1910.10683}
\showURL{%
\tempurl}


\bibitem[Ram and Gray(2012)]%
        {conetree_mips_kdd12}
\bibfield{author}{\bibinfo{person}{Parikshit Ram} {and} \bibinfo{person}{Alexander~G. Gray}.} \bibinfo{year}{2012}\natexlab{}.
\newblock \showarticletitle{Maximum Inner-Product Search Using Cone Trees}. In \bibinfo{booktitle}{\emph{Proceedings of the 18th ACM SIGKDD International Conference on Knowledge Discovery and Data Mining}} \emph{(\bibinfo{series}{KDD '12})}. \bibinfo{pages}{931–939}.
\newblock
\showISBNx{9781450314626}


\bibitem[Rendle et~al\mbox{.}(2020)]%
        {ncf_mf_goog_recsys20}
\bibfield{author}{\bibinfo{person}{Steffen Rendle}, \bibinfo{person}{Walid Krichene}, \bibinfo{person}{Li Zhang}, {and} \bibinfo{person}{John Anderson}.} \bibinfo{year}{2020}\natexlab{}.
\newblock \showarticletitle{Neural Collaborative Filtering vs. Matrix Factorization Revisited}. In \bibinfo{booktitle}{\emph{Fourteenth ACM Conference on Recommender Systems (RecSys'20)}}. \bibinfo{pages}{240–248}.
\newblock
\showISBNx{9781450375832}


\bibitem[Robertson and Zaragoza(2009)]%
        {bm25_2009}
\bibfield{author}{\bibinfo{person}{Stephen Robertson} {and} \bibinfo{person}{Hugo Zaragoza}.} \bibinfo{year}{2009}\natexlab{}.
\newblock \showarticletitle{The Probabilistic Relevance Framework: BM25 and Beyond}.
\newblock \bibinfo{journal}{\emph{Found. Trends Inf. Retr.}} \bibinfo{volume}{3}, \bibinfo{number}{4} (\bibinfo{date}{April} \bibinfo{year}{2009}), \bibinfo{pages}{333–389}.
\newblock
\showISSN{1554-0669}
\urldef\tempurl%
\url{https://doi.org/10.1561/1500000019}
\showDOI{\tempurl}


\bibitem[Santhanam et~al\mbox{.}(2022)]%
        {colbertv2_naacl22}
\bibfield{author}{\bibinfo{person}{Keshav Santhanam}, \bibinfo{person}{Omar Khattab}, \bibinfo{person}{Jon Saad-Falcon}, \bibinfo{person}{Christopher Potts}, {and} \bibinfo{person}{Matei Zaharia}.} \bibinfo{year}{2022}\natexlab{}.
\newblock \showarticletitle{{C}ol{BERT}v2: Effective and Efficient Retrieval via Lightweight Late Interaction}. In \bibinfo{booktitle}{\emph{Proceedings of the 2022 Conference of the North American Chapter of the Association for Computational Linguistics: Human Language Technologies}}, \bibfield{editor}{\bibinfo{person}{Marine Carpuat}, \bibinfo{person}{Marie-Catherine de~Marneffe}, {and} \bibinfo{person}{Ivan~Vladimir Meza~Ruiz}} (Eds.). \bibinfo{publisher}{Association for Computational Linguistics}, \bibinfo{address}{Seattle, United States}, \bibinfo{pages}{3715--3734}.
\newblock
\urldef\tempurl%
\url{https://doi.org/10.18653/v1/2022.naacl-main.272}
\showDOI{\tempurl}


\bibitem[Shazeer et~al\mbox{.}(2017)]%
        {smoe_iclr17}
\bibfield{author}{\bibinfo{person}{Noam Shazeer}, \bibinfo{person}{*Azalia Mirhoseini}, \bibinfo{person}{*Krzysztof Maziarz}, \bibinfo{person}{Andy Davis}, \bibinfo{person}{Quoc Le}, \bibinfo{person}{Geoffrey Hinton}, {and} \bibinfo{person}{Jeff Dean}.} \bibinfo{year}{2017}\natexlab{}.
\newblock \showarticletitle{Outrageously Large Neural Networks: The Sparsely-Gated Mixture-of-Experts Layer}. In \bibinfo{booktitle}{\emph{International Conference on Learning Representations}}.
\newblock
\urldef\tempurl%
\url{https://openreview.net/forum?id=B1ckMDqlg}
\showURL{%
\tempurl}


\bibitem[Shen et~al\mbox{.}(2023)]%
        {moduleformermoe_2023}
\bibfield{author}{\bibinfo{person}{Yikang Shen}, \bibinfo{person}{Zheyu Zhang}, \bibinfo{person}{Tianyou Cao}, \bibinfo{person}{Shawn Tan}, \bibinfo{person}{Zhenfang Chen}, {and} \bibinfo{person}{Chuang Gan}.} \bibinfo{year}{2023}\natexlab{}.
\newblock \bibinfo{title}{ModuleFormer: Modularity Emerges from Mixture-of-Experts}.
\newblock
\newblock
\showeprint[arxiv]{2306.04640}~[cs.CL]
\urldef\tempurl%
\url{https://arxiv.org/abs/2306.04640}
\showURL{%
\tempurl}


\bibitem[Shrivastava and Li(2014)]%
        {alsh_ping_neurips2014}
\bibfield{author}{\bibinfo{person}{Anshumali Shrivastava} {and} \bibinfo{person}{Ping Li}.} \bibinfo{year}{2014}\natexlab{}.
\newblock \showarticletitle{Asymmetric LSH (ALSH) for Sublinear Time Maximum Inner Product Search (MIPS)}. In \bibinfo{booktitle}{\emph{Advances in Neural Information Processing Systems}}, Vol.~\bibinfo{volume}{27}.
\newblock


\bibitem[Song et~al\mbox{.}(2019)]%
        {autoint_cikm19}
\bibfield{author}{\bibinfo{person}{Weiping Song}, \bibinfo{person}{Chence Shi}, \bibinfo{person}{Zhiping Xiao}, \bibinfo{person}{Zhijian Duan}, \bibinfo{person}{Yewen Xu}, \bibinfo{person}{Ming Zhang}, {and} \bibinfo{person}{Jian Tang}.} \bibinfo{year}{2019}\natexlab{}.
\newblock \showarticletitle{AutoInt: Automatic Feature Interaction Learning via Self-Attentive Neural Networks}. In \bibinfo{booktitle}{\emph{Proceedings of the 28th ACM International Conference on Information and Knowledge Management}} (Beijing, China) \emph{(\bibinfo{series}{CIKM '19})}. \bibinfo{publisher}{Association for Computing Machinery}, \bibinfo{address}{New York, NY, USA}, \bibinfo{pages}{1161–1170}.
\newblock
\showISBNx{9781450369763}
\urldef\tempurl%
\url{https://doi.org/10.1145/3357384.3357925}
\showDOI{\tempurl}


\bibitem[Sun et~al\mbox{.}(2023)]%
        {genret_neurips23}
\bibfield{author}{\bibinfo{person}{Weiwei Sun}, \bibinfo{person}{Lingyong Yan}, \bibinfo{person}{Zheng Chen}, \bibinfo{person}{Shuaiqiang Wang}, \bibinfo{person}{Haichao Zhu}, \bibinfo{person}{Pengjie Ren}, \bibinfo{person}{Zhumin Chen}, \bibinfo{person}{Dawei Yin}, \bibinfo{person}{Maarten Rijke}, {and} \bibinfo{person}{Zhaochun Ren}.} \bibinfo{year}{2023}\natexlab{}.
\newblock \showarticletitle{Learning to Tokenize for Generative Retrieval}. In \bibinfo{booktitle}{\emph{Advances in Neural Information Processing Systems}}, \bibfield{editor}{\bibinfo{person}{A.~Oh}, \bibinfo{person}{T.~Naumann}, \bibinfo{person}{A.~Globerson}, \bibinfo{person}{K.~Saenko}, \bibinfo{person}{M.~Hardt}, {and} \bibinfo{person}{S.~Levine}} (Eds.), Vol.~\bibinfo{volume}{36}. \bibinfo{publisher}{Curran Associates, Inc.}, \bibinfo{pages}{46345--46361}.
\newblock
\urldef\tempurl%
\url{https://proceedings.neurips.cc/paper_files/paper/2023/file/91228b942a4528cdae031c1b68b127e8-Paper-Conference.pdf}
\showURL{%
\tempurl}


\bibitem[Tan et~al\mbox{.}(2020)]%
        {nnm_wsdm20}
\bibfield{author}{\bibinfo{person}{Shulong Tan}, \bibinfo{person}{Zhixin Zhou}, \bibinfo{person}{Zhaozhuo Xu}, {and} \bibinfo{person}{Ping Li}.} \bibinfo{year}{2020}\natexlab{}.
\newblock \showarticletitle{Fast Item Ranking under Neural Network based Measures}. In \bibinfo{booktitle}{\emph{Proceedings of the 13th International Conference on Web Search and Data Mining}} (Houston, TX, USA) \emph{(\bibinfo{series}{WSDM '20})}. \bibinfo{publisher}{Association for Computing Machinery}, \bibinfo{address}{New York, NY, USA}, \bibinfo{pages}{591–599}.
\newblock
\showISBNx{9781450368223}
\urldef\tempurl%
\url{https://doi.org/10.1145/3336191.3371830}
\showDOI{\tempurl}


\bibitem[Tay et~al\mbox{.}(2022)]%
        {dsi_neurips22}
\bibfield{author}{\bibinfo{person}{Yi Tay}, \bibinfo{person}{Vinh~Q. Tran}, \bibinfo{person}{Mostafa Dehghani}, \bibinfo{person}{Jianmo Ni}, \bibinfo{person}{Dara Bahri}, \bibinfo{person}{Harsh Mehta}, \bibinfo{person}{Zhen Qin}, \bibinfo{person}{Kai Hui}, \bibinfo{person}{Zhe Zhao}, \bibinfo{person}{Jai Gupta}, \bibinfo{person}{Tal Schuster}, \bibinfo{person}{William~W. Cohen}, {and} \bibinfo{person}{Donald Metzler}.} \bibinfo{year}{2022}\natexlab{}.
\newblock \showarticletitle{Transformer Memory as a Differentiable Search Index}. In \bibinfo{booktitle}{\emph{Advances in Neural Information Processing Systems}}, \bibfield{editor}{\bibinfo{person}{Alice~H. Oh}, \bibinfo{person}{Alekh Agarwal}, \bibinfo{person}{Danielle Belgrave}, {and} \bibinfo{person}{Kyunghyun Cho}} (Eds.).
\newblock
\urldef\tempurl%
\url{https://openreview.net/forum?id=Vu-B0clPfq}
\showURL{%
\tempurl}


\bibitem[Wang et~al\mbox{.}(2022a)]%
        {flashlight_log22}
\bibfield{author}{\bibinfo{person}{Yiwei Wang}, \bibinfo{person}{Bryan Hooi}, \bibinfo{person}{Yozen Liu}, \bibinfo{person}{Tong Zhao}, \bibinfo{person}{Zhichun Guo}, {and} \bibinfo{person}{Neil Shah}.} \bibinfo{year}{2022}\natexlab{a}.
\newblock \showarticletitle{Flashlight: Scalable Link Prediction With Effective Decoders}. In \bibinfo{booktitle}{\emph{Proceedings of the First Learning on Graphs Conference}} \emph{(\bibinfo{series}{Proceedings of Machine Learning Research}, Vol.~\bibinfo{volume}{198})}, \bibfield{editor}{\bibinfo{person}{Bastian Rieck} {and} \bibinfo{person}{Razvan Pascanu}} (Eds.). \bibinfo{publisher}{PMLR}, \bibinfo{pages}{14:1--14:17}.
\newblock
\urldef\tempurl%
\url{https://proceedings.mlr.press/v198/wang22a.html}
\showURL{%
\tempurl}


\bibitem[Wang et~al\mbox{.}(2022b)]%
        {nci_neurips22}
\bibfield{author}{\bibinfo{person}{Yujing Wang}, \bibinfo{person}{Yingyan Hou}, \bibinfo{person}{Haonan Wang}, \bibinfo{person}{Ziming Miao}, \bibinfo{person}{Shibin Wu}, \bibinfo{person}{Qi Chen}, \bibinfo{person}{Yuqing Xia}, \bibinfo{person}{Chengmin Chi}, \bibinfo{person}{Guoshuai Zhao}, \bibinfo{person}{Zheng Liu}, \bibinfo{person}{Xing Xie}, \bibinfo{person}{Hao Sun}, \bibinfo{person}{Weiwei Deng}, \bibinfo{person}{Qi Zhang}, {and} \bibinfo{person}{Mao Yang}.} \bibinfo{year}{2022}\natexlab{b}.
\newblock \showarticletitle{A Neural Corpus Indexer for Document Retrieval}. In \bibinfo{booktitle}{\emph{Advances in Neural Information Processing Systems}}, \bibfield{editor}{\bibinfo{person}{S.~Koyejo}, \bibinfo{person}{S.~Mohamed}, \bibinfo{person}{A.~Agarwal}, \bibinfo{person}{D.~Belgrave}, \bibinfo{person}{K.~Cho}, {and} \bibinfo{person}{A.~Oh}} (Eds.), Vol.~\bibinfo{volume}{35}. \bibinfo{publisher}{Curran Associates, Inc.}, \bibinfo{pages}{25600--25614}.
\newblock
\urldef\tempurl%
\url{https://proceedings.neurips.cc/paper_files/paper/2022/file/a46156bd3579c3b268108ea6aca71d13-Paper-Conference.pdf}
\showURL{%
\tempurl}


\bibitem[Yang et~al\mbox{.}(2018)]%
        {mos_iclr18}
\bibfield{author}{\bibinfo{person}{Zhilin Yang}, \bibinfo{person}{Zihang Dai}, \bibinfo{person}{Ruslan Salakhutdinov}, {and} \bibinfo{person}{William~W. Cohen}.} \bibinfo{year}{2018}\natexlab{}.
\newblock \showarticletitle{Breaking the Softmax Bottleneck: A High-Rank {RNN} Language Model}. In \bibinfo{booktitle}{\emph{International Conference on Learning Representations (ICLR'18)}}.
\newblock


\bibitem[Zhai et~al\mbox{.}(2023)]%
        {zhai23kdd}
\bibfield{author}{\bibinfo{person}{Jiaqi Zhai}, \bibinfo{person}{Zhaojie Gong}, \bibinfo{person}{Yueming Wang}, \bibinfo{person}{Xiao Sun}, \bibinfo{person}{Zheng Yan}, \bibinfo{person}{Fu Li}, {and} \bibinfo{person}{Xing Liu}.} \bibinfo{year}{2023}\natexlab{}.
\newblock \showarticletitle{Revisiting Neural Retrieval on Accelerators}. In \bibinfo{booktitle}{\emph{Proceedings of the 29th ACM SIGKDD Conference on Knowledge Discovery and Data Mining}} (Long Beach, CA, USA) \emph{(\bibinfo{series}{KDD '23})}. \bibinfo{publisher}{Association for Computing Machinery}, \bibinfo{address}{New York, NY, USA}, \bibinfo{pages}{5520–5531}.
\newblock
\showISBNx{9798400701030}
\urldef\tempurl%
\url{https://doi.org/10.1145/3580305.3599897}
\showDOI{\tempurl}


\bibitem[Zhai et~al\mbox{.}(2024)]%
        {zhai2024actions_icml24}
\bibfield{author}{\bibinfo{person}{Jiaqi Zhai}, \bibinfo{person}{Lucy Liao}, \bibinfo{person}{Xing Liu}, \bibinfo{person}{Yueming Wang}, \bibinfo{person}{Rui Li}, \bibinfo{person}{Xuan Cao}, \bibinfo{person}{Leon Gao}, \bibinfo{person}{Zhaojie Gong}, \bibinfo{person}{Fangda Gu}, \bibinfo{person}{Jiayuan He}, \bibinfo{person}{Yinghai Lu}, {and} \bibinfo{person}{Yu Shi}.} \bibinfo{year}{2024}\natexlab{}.
\newblock \showarticletitle{Actions Speak Louder than Words: Trillion-Parameter Sequential Transducers for Generative Recommendations}. In \bibinfo{booktitle}{\emph{Proceedings of the 41st International Conference on Machine Learning}} \emph{(\bibinfo{series}{Proceedings of Machine Learning Research}, Vol.~\bibinfo{volume}{235})}, \bibfield{editor}{\bibinfo{person}{Ruslan Salakhutdinov}, \bibinfo{person}{Zico Kolter}, \bibinfo{person}{Katherine Heller}, \bibinfo{person}{Adrian Weller}, \bibinfo{person}{Nuria Oliver}, \bibinfo{person}{Jonathan Scarlett}, {and} \bibinfo{person}{Felix Berkenkamp}} (Eds.). \bibinfo{publisher}{PMLR}, \bibinfo{pages}{58484--58509}.
\newblock
\urldef\tempurl%
\url{https://proceedings.mlr.press/v235/zhai24a.html}
\showURL{%
\tempurl}


\bibitem[Zhai et~al\mbox{.}(2011)]%
        {hdss_sigmod11}
\bibfield{author}{\bibinfo{person}{Jiaqi Zhai}, \bibinfo{person}{Yin Lou}, {and} \bibinfo{person}{Johannes Gehrke}.} \bibinfo{year}{2011}\natexlab{}.
\newblock \showarticletitle{ATLAS: A Probabilistic Algorithm for High Dimensional Similarity Search}. In \bibinfo{booktitle}{\emph{Proceedings of the 2011 ACM SIGMOD International Conference on Management of Data}} \emph{(\bibinfo{series}{SIGMOD '11})}. \bibinfo{pages}{997–1008}.
\newblock
\showISBNx{9781450306614}


\bibitem[Zhu et~al\mbox{.}(2018)]%
        {tdm_kdd18}
\bibfield{author}{\bibinfo{person}{Han Zhu}, \bibinfo{person}{Xiang Li}, \bibinfo{person}{Pengye Zhang}, \bibinfo{person}{Guozheng Li}, \bibinfo{person}{Jie He}, \bibinfo{person}{Han Li}, {and} \bibinfo{person}{Kun Gai}.} \bibinfo{year}{2018}\natexlab{}.
\newblock \showarticletitle{Learning Tree-Based Deep Model for Recommender Systems}. In \bibinfo{booktitle}{\emph{Proceedings of the 24th ACM SIGKDD International Conference on Knowledge Discovery and Data Mining}} (London, United Kingdom) \emph{(\bibinfo{series}{KDD '18})}. \bibinfo{pages}{1079–1088}.
\newblock
\showISBNx{9781450355520}


\bibitem[Zhuang et~al\mbox{.}(2023)]%
        {dsiqg_zhuang2023bridginggapindexingretrieval}
\bibfield{author}{\bibinfo{person}{Shengyao Zhuang}, \bibinfo{person}{Houxing Ren}, \bibinfo{person}{Linjun Shou}, \bibinfo{person}{Jian Pei}, \bibinfo{person}{Ming Gong}, \bibinfo{person}{Guido Zuccon}, {and} \bibinfo{person}{Daxin Jiang}.} \bibinfo{year}{2023}\natexlab{}.
\newblock \bibinfo{title}{Bridging the Gap Between Indexing and Retrieval for Differentiable Search Index with Query Generation}.
\newblock
\newblock
\showeprint[arxiv]{2206.10128}~[cs.IR]
\urldef\tempurl%
\url{https://arxiv.org/abs/2206.10128}
\showURL{%
\tempurl}


\bibitem[Zhuo et~al\mbox{.}(2020)]%
        {otm_icml20}
\bibfield{author}{\bibinfo{person}{Jingwei Zhuo}, \bibinfo{person}{Ziru Xu}, \bibinfo{person}{Wei Dai}, \bibinfo{person}{Han Zhu}, \bibinfo{person}{Han Li}, \bibinfo{person}{Jian Xu}, {and} \bibinfo{person}{Kun Gai}.} \bibinfo{year}{2020}\natexlab{}.
\newblock \showarticletitle{Learning optimal tree models under beam search}. In \bibinfo{booktitle}{\emph{Proceedings of the 37th International Conference on Machine Learning}} \emph{(\bibinfo{series}{ICML'20})}. \bibinfo{publisher}{JMLR.org}, Article \bibinfo{articleno}{1080}, \bibinfo{numpages}{10}~pages.
\newblock


\end{thebibliography}
